\documentclass[12pt]{article}
\usepackage{natbib}
\usepackage{amsmath,amsthm,amssymb}
\usepackage[latin1]{inputenc}
\usepackage{epsfig}
\usepackage{color}
\usepackage{dsfont}
\usepackage{amssymb,caption}
\usepackage{graphics}
\newcommand{\boldtheta}{\mbox{\boldmath $\theta$}}

\newtheorem{theorem}{\bf Theorem}

\newtheorem{prop}{\bf Proposition}

\newcommand{\ep}{\varepsilon}
\newcommand{\al}{\alpha}

\newcommand{\la}{\lambda}

\newcommand{\ri}{\rightarrow}

\usepackage{rotating} %% pacote para girar as tabelas
\usepackage{rcs}

\title{\Large \bf   Seasonal fractional long-memory processes. A semiparametric estimation approach}
\author{\small VALDERIO A. REISEN$^{a}$\thanks{Corresponding author. E-mail: valderioanselmoreisen@gmail.com}
,    \small WILFREDO PALMA$^b$,   \small JOSU ARTECHE$^c$ \\ and \small BARTOLOMEU  ZAMPROGNO$^a$\\
\small $^a${\em Department of Statistics, CCE and PPGEA-CT-UFES, Vitoria - ES,
Brazil.}\\ \small  $^b${\em Department of Statistics, PUC-Chile} \\\small $^c${\em Department of Econometrics and  Statistics, Universidad del Pais Vasco-Spain}}
\date{\today}

\begin{document}
\maketitle

\begin{abstract}
This paper explores seasonal and long-memory time series properties by using the seasonal fractional ARIMA model when the seasonal data has one and two seasonal
periods and short-memory counterparts. The stationarity and invertibility parameter conditions are established for the model
studied. To  estimate the memory parameters, the method given in Reisen, Rodrigues and Palma (2006 a,b) is  generalized here to deal with a time series
with two seasonal fractional  long-memory parameters. The asymptotic properties are established and the accuracy of the method is investigated through
Monte Carlo experiments. The good performance of the estimator indicates that it can be an alternative competitive  procedure to  estimate
 seasonal  long-memory time series data.  Artificial and PM$_{10}$ series were considered as examples of applications of the proposed estimation method. \\

\end{abstract}

\section{Introduction}

Time series exhibiting seasonal or cyclical characteristics are
very common in economics, hydrology, and many other disciplines.
As a consequence, several methodologies have been developed to
deal with these features. One of the most well known of these
tools is the class of seasonal autoregressive integrated moving
average (SARIMA) process. This model can describe many time series
containing a mixture of seasonal phenomena of different periods. It is well
known that many series may contain a persistent seasonal structure along with a long run trend.  However, the case of more than one seasonal (non-zero) structures  has not very much studied. In this direction,  the study of the ARFIMA process with seasonal periods becomes an interesting research topic, and this is the main motivation of this work.

Let $X_t$ $\equiv$ $\{X_t\}_{t \in \mathds{Z}}$  be a time series with a zero mean   and a  constant variance. A multiple seasonal ARIMA model can be written as follows:
\begin{equation}
\label{model1} \prod_{j=1}^M\phi_j(B)\nabla^{\textbf{d}}(B)X_t
=\prod_{\ell=1}^N\theta_\ell(B)\varepsilon_t,
\end{equation}
where $\nabla^{\textbf{d}}(B) \equiv \prod _{i=1}^k (1-B^{s_i})^{d_i}$,
 $M$ and $N$ are, respectively, the number of factors of the AR and MA components,
 $d{_i} \in \mathds{N}$ , $i= 1,...,k$, is the differencing
parameter  and $k$ is the number of differencing  factors, $s_i$
is the $i$-seasonal period, $BX_t$ = $X_{t-1}$, and
$\varepsilon_t$ is a white noise process with zero-mean and
variance
 $\sigma_{\varepsilon}^2$. $\phi_{j}(B),j=1,...,M,$ and $\theta_{\ell}(B),\ell=1,...,N,$ can
also be polynomials with seasonal effects. The stationarity and
invertibility properties of model (\ref{model1}) are established based on certain parameter
conditions. The  multiple seasonal ARIMA model belongs to a class of models with a general
difference operator given by
\begin{equation}\label{model2}
\nabla^{\textbf{d}}(B)\equiv\prod^L_{\ell=1}[(1-Be^{i\lambda_{\ell}})(1-Be^{-i\lambda_{\ell}})]^{d_{\ell}},
\end{equation}
where now $d{_{\ell}} \in \mathds{R}$ ($d{_{\ell}} > -1$) and it
is defined as the {\it fractionally differencing parameter} and $
\lambda_{\ell}$, $\ell =1,...,L$,
 are fixed frequencies in the range $[-\pi,\pi]$.  For suitable choices of the fractional parameters, time series models with filter given in (\ref{model2})   may have  a finite number of zeros or
singularities of order $d_{1}$,...,$d_{L}$ on the unit circle which allows the modeling of
long and short memory data containing seasonal periodicities. In the time domain, the usual definition of long memory
is the non-summability

\begin{equation*}
\sum_{h=0}^{\infty}|\gamma(h)|= \infty,
\end{equation*}
where $\gamma(h)$ is the autocovariance at lag $h$ of the process, whereas, in the frequency domain, this property is defined by the fact that the spectral density of the process becomes unbounded at some  frequency in [0,$\pi$].

A  time series with both seasonal and non-seasonal fractional
differencing parameters  has a spectral density specified  by
\begin{equation}
\label{model3}
f(\lambda) = f^{*}(\lambda)|\lambda|^{-2d}\prod^k_{i=1}\prod^{\xi_i}_{j=1}|\lambda -\lambda_{ij}|^{-2d_{i}},
\end{equation}
where $d{_{i}} \in \mathds{R}$ ($d{_{i}} > -1$), $\lambda \in (-\pi, \pi]$, $f^{*}(\lambda)$ is a continuous
function,   bounded above  and away from zero and
$\lambda_{ij}\neq 0$ are poles for $j =1,..., \xi_i$, $i=1,...,k$.
Processes  with a spectral density given by \eqref{model3} have
been discussed by Arteche and Robinson (1999), Giraitis and Leipus (1995), Leipus and Viano (2000),  Palma and Chan (2005) and Palma (2007, Ch.12), among
others, to model time series with seasonal and cyclical
long-memory behavior.

The main interest in models which have filter  (\ref{model2})
 and spectral density of the form (\ref{model3}) is related to the estimation of  fractional memory parameters
$d_{1}$,...,$d_L$. Ray (1993) used IBM product revenues to
illustrate the usefulness of modeling seasonal fractionally differenced ARMA models by allowing  two seasonal
fractional differencing parameters in the model, one at lag three and the other at lag twelve.
%The fractional estimates were
%obtained by using a non-linear least square method under a parameter restriction.
Other papers related to this topic are, for example,  Hassler (1994), Gray et al.
(1989,1994), Giraitis and Leipus (1995),   Ooms and France (2001), Woodward et al. (1998), Arteche and Robinson (2000), Palma and Chan (2005) and Reisen et
al. (2006a,b), Gil-Alana (2001) among others. Hassler (1994)  introduced the rigid
and flexible filters and an application of this methodology to the
economic activities in the Euro area is  discussed in Ferrara and
Guegan (2006). Arteche and Robinson (2000), Arteche (2002) and
Arteche and Velasco (2005)  dealt with robust semiparametric
estimators and testing procedures for the seasonal fractional
memory parameter. Woodward et al. (1998) extended the Gegenbauer
ARMA process (GARMA). %In turn, this class
%of models was extended by Gray et al. (1989, 1994) to a $k$-factor
%GARMA process.
Independently of these works, a time series model
for fitting long or short-memory data containing seasonal
periodicities was introduced by Giraitis and Leipus (1995) which
is called  the Fractionally Autoregressive Unit Circle Moving
Average model (ARUMA). These authors discussed the asymptotic
properties of the ARUMA model and  the estimation of its
parameters.

 Another equally relevant publication related to the asymptotic properties of
seasonal and periodic time series is the work by Viano et al. (1995). Reisen et al.
(2006a,b) dealt with the estimation of the seasonal ARFIMA model with long-memory
innovations  (SARFIMA $(0,d,0)\times (0,D,0)_s$) by using different estimation
procedures for the seasonal and non-seasonal memory parameters, that is, for $D$ and $d$
respectively. The estimators are based on the multilinear regression equation of  $\log
f(\cdot)$, where $f(\cdot)$ is the spectral density of the process  satisfying
$f(\lambda) \sim C^*|\lambda|^{-2(d+D)}$ as $\lambda \to 0$,  where $C^*$ is a positive
constant. Necessary conditions that guarantee the stationary and invertibility of the
model were also established. Through Monte Carlo experiments, they compared their
proposed methodology with other well-known parametric estimation procedures  such as the
Whittle and the maximum likelihood methods. The empirical evidence showed that the
multilinear regression estimators are very promising.

Most of the works referred to above deal with  the estimation of one seasonal
long-memory parameter. However, in many practical situations
the time series exhibits more than one seasonal component. In
order to explore these more complex situations, this paper focuses
on the estimation of models containing one and two seasonal
periods which encompass long and short-memory dependence structures. Specifically,  an ordinary least squares (OLS) procedure, based on a log-periodogram regression,
is proposed  to estimate all fractional  parameters simultaneously.

Let now  $X_t$ $\equiv$ $\{X_t\}_{t \in \mathds{Z}}$ be a zero-mean time series defined by
\begin{equation} \label{model4}
 \nabla^{\textbf{d}}(B)X_t\equiv(1-B^{s_1})^{d_1}(1-B^{s_2})^{d_2}X_t=\nu_t,
\end{equation}
where  the vector $\textbf{d}$ =$ (d_1,d_2)^{'}$, $s_1$ and $s_2$ are seasonal
periods and $d_1, d_2 \in \mathds{R}$ ($d{_{i}} > -1$) are their seasonal memory parameters,
respectively, and $\nu_t$ has a spectral density that satisfies
the following assumption.
%is a Gaussian stationary covariance ARMA process, e.g. $\nu_t=\psi_{\nu}(B)\varepsilon_t$.

{\bf Assumption 1:} The spectral density of $\nu_t $ satisfies as $\lambda \rightarrow 0$
\[f_{\nu}\left(\frac{2\pi k}{s'}+\lambda \right)=f_{\nu,k}+c_k|\lambda|^{\alpha_k}+O(|\lambda|^{\alpha_k +\iota})\]
for some $\iota>0$, $f_{\nu,k}$ $\equiv$ $f_\nu\left(\frac{2\pi k}{s'}\right)$, $k=0,1,...,[s'/2]$,  $s'=\max (s_1,s_2)= s_1 $ (without loss of generality) and $\alpha_k=\alpha_1$ for $k=0,s'/2$ (if $s'$ even) and $\alpha_k=\alpha_2$ otherwise.  If $\nu_t$ is a stationary and invertible ARMA process then $\al_1=2$, $\al_2=1$ and $\iota=1$. In this case, the process $X_t$ is usually defined as Seasonal ARFIMA (SARFIMA) model.

In the next section  some properties of the model given by \eqref{model4} are
discussed. In  particular, the   stationarity and invertibility conditions of Model \eqref{model4}  are established in Proposition 1. The estimation of these models is
discussed in Section 3, where the proposed  ordinary least squares (OLS) estimator is introduced. Some
asymptotic properties of these estimators are established in Theorems 1 and 2. For
example, the proposed OLS estimator is shown to be  asymptotically unbiased and
normally distributed. For comparison purposes,  the quasi-likelihood Fox-Taqqu
estimator (Fox and Taqqu (1986))  is  adapted here for Gaussian seasonal long-memory processes. The comparison between parametric and semiparametric approaches may appear to be unfair for the former class, in the case of a correct and complete parametric specification.  So,  the misspecification problem of the FT method is also included to be a part of the simulation section.  The finite sample performance of the proposed estimator is investigated in Section 4 while Section 5 discuss some applications. Final remarks are presented in Section 6.

\section{Model properties}
Let $X_t$ be a time series process defined by
(\ref{model4}). For simplicity, it is  assumed   that $s_1$ and $s_2$ are even
numbers.
 %This is the most common case  in many sciences such as  in Economics where the series usually have seasonal periods  of $s=4$ (quarterly data) and $s=12$ ( monthly series).
The fractional $d_i$ difference is a generalization of the binomial
expression $(1-B)^d$ and it can be written as
\begin{eqnarray*}
(1-B^{s_i})^{d_i}= \sum_{k=0}^\infty \pi_k B^{ks_i},
\end{eqnarray*}
where
\begin{eqnarray*}
\pi_0 = 1, \qquad \pi _k = \frac{\Gamma (k-d_i)}{\Gamma ( k+1) \Gamma (-d_i)}, \quad
i=1,2,
\end{eqnarray*}
  and $\Gamma (\cdot)$ is the Gamma function.

 In the literature of the seasonal long-memory process, there are some specific time
series models of interest  obtained from the solution of the
general fractional operator (\ref{model2}) and the spectral
density of form \eqref{model3}. The specific filters and their
models are:
 (a) $(1-B)^d$  is the filter of the fractional integrated I($d$) process see, for
example, Hosking (1981), among others; (b) $(1-B)^{d_{1}}(1-B^s)^{d_{2}}$ is the filter in the SARFIMA process
that has been  explored in the literature of seasonal fractional ARMA model, see for example, Porter-Hudak (1990), Hassler (1993), Arteche (2002), Arteche and Robinson (2000) and Reisen et al. (2006a,b); (c)
$(1-\upsilon_1B-\cdots-\upsilon_vB^d)$  is the filter that belongs to the
ARUMA and  $k$-GARMA  processes, proposed  by Giraitis and Leipus (1995)  and
independently by Woodward et al. (1998), respectively;
(d) $(1-B^3)^{d_{1}}(1-B^{12})^{d_{2}}$ is the filter used by Ray (1993)
to model and forecast a monthly IBM revenue data under the restriction $d_1 +d_2
=1$.

Returning to our specific  model  of interest for $X_t$ given in \eqref{model4}, the
filter may be written as follows:
\begin{equation} \label{model5}
(1-B^{s_1})^{d_1}(1-B^{s_2})^{d_2}=\prod^2_{i=1}\prod^{\xi_i}_{j=0}
[{(1-Be^{i\lambda_{ij}})(1-Be^{-i\lambda_{ij}})}]^{d_{ij}},
\end{equation}
where $\lambda_{ij}=\frac{2\pi j}{s_i}$  $(j=0, 1,\ldots \xi_i)$
are the frequencies of the period $i$, $\xi_i= \frac{s_i}{2}$
($i=1,2$), and

% SOLUTION JOSU..I HAVE TO CHECK $$$$$$$$$$$$$$$$$$$$$$$$$$$$$$$$
 %\textcolor{red}{(I don't see why we impose equality of $d_{1j}$ and $d_{2j}$ when both seasonal components share the same frequency. Would not it be better and less restrictive to define}
%
%\textcolor{red}{

 \begin{eqnarray}
 \label{model6}
             \begin{array}{ll}
              d_{1j}=d_1 & \mbox{ when $\lambda_{1j}\neq \lambda_{2j},0, \pi$;}\\
              d_{1j}=d_1/2 & \mbox{ when $\lambda_{1j}\neq \lambda_{2j}$ and $\lambda_{1j}=0, \pi$;}\\
              d_{1j}+d_{2j}=\frac {d_{1} + d_{2}}{2} & \mbox{ when $\lambda_{1j}= \lambda_{2j} =0, \pi$;}\\
               d_{1j}+d_{2j}= d_{1} + d_{2} & \mbox{ when $\lambda_{1j}= \lambda_{2j} \neq 0, \pi$;} \end{array}
\end{eqnarray}
and similarly when $i=2$
%$$$$$$$$$$$$$$$$$$$$$$$$$$$$$$$$$$$$$$$$$$$$$$$$$$$$$$$$$$$$$$$$

 %\begin{eqnarray}
% \label{model6}
%d_{ij}=
%   \left\{
%             \begin{array}{ll}
%              \frac {d_{1} + d_{2}}{2} & \mbox{ when $\lambda_{1j}= \lambda_{2j} =0, \pi$;}\\
%               (d_{1} + d_{2}) & \mbox{when  $0<\lambda_{1j} \equiv  \lambda_{2j'}\mod(2\pi)<\pi$;}\\
%%                 (d_{1} + d_{2}) & \mbox{when  $0<\lambda_{1j}= \lambda_{2j'}\,mod(2\pi)<\pi$;}\\
%              d_{i}            & \mbox{ at $\lambda_{ij}$, $ 0<\lambda_{ij} \neq \lambda_{i'j'}\ \mod(2\pi)<\pi$.}
%             \end{array}
%\right.
%\end{eqnarray}

  It is easy to show that the filter \eqref{model5} is a particular case of the
operator  \eqref{model2} by  using the equality

 \begin{eqnarray*}
%1-z^s=\begin{cases}
%          (1-z)(1+z)\prod\limits_{1}^{\frac{s}{2}-1}(1-z\rme ^{2\pi \rmi k/s})(1-ze^{2\pi \textcolor{red}{-i}k/s}), & s \quad \text{even}; \\
%          (1-z)\prod\limits_{1}^{\frac{s}{2}-1}(1-ze^{2\pi ik/s})(1-ze^{2\pi ik/s}), & s \quad \text{odd}.
%      \end{cases}
1-z^s= (1-z)(1+z)\prod\limits_{1}^{\frac{s}{2}-1}(1-z e ^{2\pi i k/s})(1-z e^{-2\pi i k/s}), \quad \text{for} \quad s \quad \text{even}. \\
\end{eqnarray*}
 When $s$ is an odd number, the term $(1+z)$ does not appear in the above equation.  From the expression of $1-z^s$, the following proposition is reached:
% PROPOSITION 1
\begin{prop}
Let the process $X_t$  be a solution of  equation
\begin{equation} \label{model7}
X_t=(1-B^{s_1})^{-d_1}(1-B^{s_2})^{-d_2}\nu_t,
\end{equation}
where $\nu_t$ is a covariance stationary   ARMA process  ( $\nu_t = \frac{\Theta(B)}{\Phi(B)}\epsilon_{t})$, $\epsilon_{t}$ is an i.i.d Gaussian sequence with zero mean and variance $\sigma_{\epsilon}^2$ , and $d_i \in \mathds{R}$   is the fractional
parameter at seasonal period \(s_i\) for $i=1,2$. Then,
 \begin{itemize}
\item[(a)] The process $X_t$  is stationary and invertible if \(|d_1+d_2|<1/2\) and $|d_i|<1/2$, $i
=1,2$.
 \item[(b)] The spectral density of $X_t$  is given by
\begin{eqnarray*}
\begin{array}{ll}
f(\lambda) &= f_\nu(\lambda)\prod^2_{i=1}\prod^{\xi_i}_{j=0}|2
\sin (\frac{\lambda - \lambda_{ij}}{2})2\sin
(\frac{\lambda+\lambda_{ij}}{2})|^{-2d_{ij}}\\
%           &= \frac{\sigma^2}{2\pi}\prod^2_{j=1} |2 \cos \lambda - \cos \lambda_j|^{-2d_j}
           &= f_\nu(\lambda)\Big(2\sin \frac{\lambda s_1}{2}\Big)^{-2d_1}\Big(2\sin \frac{\lambda
           s_2}{2}\Big)^{-2d_2},
\end{array}
\end{eqnarray*}
where $f_\nu(\lambda)$ ( $ 0 \leq \lambda \leq \pi$) is the
spectral density of \(\nu_t\),  $\lambda_{ij}=\frac{2\pi
j}{s_i}$, $i=1,2$ and $j=0,1,...,\frac{s_i}{2}$, and $d_{ij}$
are given by (\ref{model6}).
\item[(c)] Assuming that  $\max\{d_{ij}\}>0$, the asymptotic autocovariance of  $X_t$,  $\gamma(h)= \mathds{E}(X_h X_0)$,
is given by
\begin{eqnarray*}
\gamma(h)=\sum_{i=1}^{2} \sum_{j=1}^{\xi_i}{a_{ij}}\mid h \mid^{2d_{ij}-1}(\cos h\lambda_{ij} +
o(1)) \quad\text{as} \quad k \to \infty,
\end{eqnarray*}
\noindent where
\begin{eqnarray*}
    a_{ij}=\begin{cases}
a_{ij}^{'}  & \lambda_{ij} =0,\pi \\
2a_{ij}^{'} & 0<\lambda_{ij} <\pi,
\end{cases}
\end{eqnarray*}
$d_{ij}$ is specified as in (\ref{model6})  and

\begin{eqnarray*}
a_{ij}^{'}=\mid\
\frac{\Theta(e^{-2\pi \lambda_{ij}})}{\Phi(e^{-2\pi \lambda_{ij}})}\mid^{2}
\frac{\sigma_{\epsilon}^2}{\pi}\Gamma(1-2d_{ij})\sin(d_{ij}\pi)D^2_{ij},
\end{eqnarray*}

 where
\begin{eqnarray*}
D_{ij}=\begin{cases}
\mid 2\sin\lambda_{ij}\mid^{-d_{ij}}\prod_{\ell \neq j }\mid 2(\cos\lambda_{ij}-\cos\lambda_{i\ell})
\mid^{-d_{i\ell }}, & 0<\lambda_{ij}<\pi, \\
\prod_{\ell\neq j}\mid 2(\cos\lambda_{ij}-\cos\lambda_{i\ell})\mid^{-d_{i\ell}}, &
\lambda_{ij}=0,\pi.
           \end{cases}
\end{eqnarray*}
\end{itemize}
\end{prop}
\begin{proof}  (a)  As previously noted,   filter \eqref{model5} is a particular
case of \eqref{model2} which is the operator of the ARUMA$(p,d_1,...,d_L,q)$ model
where $p$ and $q$ are the polynomial orders of a stationary and invertible
ARMA$(p,q)$ process. From Theorem 1 of Giraitis and Leipus (1995), the ARUMA
$(0,d_1,...,d_L,0)$ process is stationary and invertible if the fractional
parameters $d_{\ell}$, $\ell=1,...,L$, in \eqref{model2} satisfy $|d_{\ell}| < 1/2$ when
$\lambda_{\ell} \neq 0,\pi$ and $|d_{\ell}| < 1/4$ otherwise. From this fact and by means
of equations \eqref{model5} and \eqref{model6}, it is straightforward to establish
the stationary and invertibility properties. The proof of (b) is immediately
obtained from ($\ref{model3}$) and Theorem 2 in Giraitis and Leipus (1995). This
theorem is also used to prove the asymptotic covariance given in ($c$) where
$d_{ij}$  is defined by (\ref{model6}).

\end{proof}

%\begin{rem}
%The following two facts can be readily deduced from Proposition 1.
%\begin{itemize}
%\item[(a)] As $\lambda \to \lambda_{ij}, \ \ f( \lambda ) \sim
%C|\lambda-\lambda_{ij}|^{-2\alpha_i}$, where C is a constant  and $ {\alpha_i} = d_i$ $\forall$ $ 0<\lambda_{ij} \neq \lambda_{i'j'}\mod(2\pi)<\pi$
%or  $ {\alpha_i} = d_1 + d_2$ otherwise.
%
%\item[(b)]  The process is stationary and long memory  only for parameters $d_1$ and $d_2$ satisfying,
% $0<d_1 + d_2< \frac{1}{2}$  and
%  $ 0<d_1, d_2< \frac{1}{2}$.
%\end{itemize}
%\end{rem}

\section{ Seasonal fractional parameter estimators}
This section deals with the  estimation method based on the
regression equation of  $ \log f(\lambda)$  to obtain the
estimates of  model (\ref{model7}). Since the procedure  proposed
here provides simultaneous estimates for  multiple seasonal memory
parameters, the method is a more general approach  than those
discussed in Reisen et al. (2006a,b) and related references.  Let $n$ be the sample
size and let $X_1,\dots,X_n$ be a realization of the process defined
by \eqref{model7}, where $\nu_t $ is a Gaussian ARMA process. The well-known
periodogram function   $I(\lambda) = {(2 \pi n)}^{-1} |\sum_{t=1}^{n}
X_{t}e^{i\lambda t}|^{2}$ is an asymptotic unbiased and inconsistent estimator the spectral density and it is the standard estimator used in  time series modeling.
%\textcolor{red}{I would include the assumptions just before Theorem 1, and in any case after eq. (8) that defines the bandwidth $m$}

%{\bf Assumption 2:} $s_1$ is a multiple of $s_2$.
%
%{\bf Assumption 3:} Let $m$ and $n$ numbers that satisfy
%\[\left(\frac{m}{n}\right)^{\iota}\log m +\frac{1}{m}\rightarrow
%0\;\;\;\mbox{ as }\; n\rightarrow \infty\] for some $\iota>0$.
%
%{\bf Assumption 4:} \[\frac{m^{\alpha^* +0.5}}{n^{\alpha^*}}\rightarrow K \;\;\;\mbox{ as }\; n\rightarrow \infty\]
%where $\alpha^*=\min (\alpha_1,\alpha_2)$.

\subsection{The OLS regression estimators}

%\textit{GPHT method}

\vspace{0.3cm} The fractional memory {\bf OLS} estimators are the slope estimators of the multiple regression equation

\begin{equation}
\log I_{k,j}=a_k-2d_1\log X_{1,k,j}-2d_2\log X_{2,k,j}+V_{k,j}\;\;\;,k=0,1,...,[s'/2], \label{model8}
\end{equation}

\noindent where $s'=\max (s_1,s_2)$, $j=1,...,m$  ( $m \in \mathds{N^*}$) if $k=0$, $j=-1,...,-m$ if $k=s'/2$ ($s$ even) and $j=\pm 1,...,\pm m$  otherwise,
$I_{k,j}=I(2\pi k/s'+\lambda_j)$, $\lambda_j=2\pi j/n$ is the Fourier frequency, $[\cdot]$ means the integer part and
\begin{eqnarray*}
a_k&=&\log f_{\nu,k}+E\left(\log\frac{ I_{k,j}}{f_{k,j}}\right)\\
V_{k,j}&=& U_{k,j}+ \varepsilon_{k,j}\\
U_{k,j}&=&\log\frac{ I_{k,j}}{f_{k,j}}-E\left(\log\frac{ I_{k,j}}{f_{k,j}}\right)\\
\varepsilon_{k,j}&=& \log \frac{ f_{\nu}\left(\frac{2\pi k}{s'}+\lambda_j\right)}{ f_{\nu,k}}
= b_k\lambda_j^{\alpha}+O(\lambda_j^{\alpha+\iota})\\
X_{1,k,j}&=&2\sin\left(\frac{s_1}{2}\left[\frac{2\pi k}{s'}+\lambda_j\right]\right)\\
X_{2,k,j}&=&2\sin\left(\frac{s_2}{2}\left[\frac{2\pi k}{s'}+\lambda_j\right]\right)
\end{eqnarray*}
for $f_{k,j}=f\left(\frac{2\pi k}{s'}+\lambda_j\right)$ and $b_k=c_k/f_{v,k}$.  The regression equation (\ref{model8}) is easily  derived from the expression of the  $\log
f(\lambda)$ where $f(\lambda)$ is the spectral density given in Proposition 1.
To avoid the estimation of the constants $a_k$, the variables are locally centered such that the estimates  are obtained by least squares in the regression model
\begin{equation}
Y_{k,j}=d_1Z_{1,k,j}+d_2Z_{2,k,j}+V^*_{k,j},\;\;\,
\label{model9}
\end{equation}
%for $j=1,...,m$ if $k=0$, $j=-1,...,-m$ if $k=s/2$ and $j=\pm 1,...,\pm m$ otherwise,
where $V^*_{kj}=V_{k,j}-\frac{1}{m_k}\sum_j^*V_{k,j}$ for $m_k=\delta_k m$ with
$\delta_k=1$ for $k=0,s'/2$ and $\delta_k=2$ otherwise and the sum $\sum^*$ runs for
$j=1,...,m$ if $k=0$, $j=-1,...,-m$ if $k=s'/2$ and $j=\pm 1,...,\pm m$ otherwise.
$Y_{k,j}$, $Z_{1,k,j}$ and $Z_{2,k,j}$ are the locally centered dependent variable
and regressors in (\ref{model8}) similarly defined. The local centering is needed
here because the regression model in (\ref{model8}) has different constants
depending on the  frequency bandwidth. A global centering can be used only if
$a_1=...=a_{[s'/2]}$ which holds for example if $\nu_t$ is  a white noise process
with a constant spectral density function.

The  estimation procedure based on the above regression equation
is  motivated by the pioneer regression estimator proposed by
Geweke and Porter-Hudak (1983) for the  ARFIMA model.
Since the introduction of the method, it has became one of the most
popular estimation procedures and its empirical and asymptotic
properties have been well established. Robinson (1995) and  Hurvich,
Deo and Brodsky (1998)  proved that  the GPH-estimator is
consistent and asymptotically normal for Gaussian time series
processes.  Hurvich et al. (1998) also established that the optimal bandwidth  is of order $O(n^{4/5})$.

%The regression estimation in the
%non-stationary region is mainly addressed by Velasco (1999), and
%the recent work by Phillips (2007) shows the consistency and
%derives the asymptotic distribution of the estimator for unit root
%processes.

When the model is a SARFIMA$(0,d,0)\times(0,d_s,0)_s$ process, Reisen et al. (2006a,b)
proposed  different estimation methods for $d_s$ and $d$. Basically, the regression
estimators considered in their study are distinguished by the choice of the
bandwidth  when regressing $\log[I(\lambda)]$ on $\log[{2 \sin(\lambda s/2)}]$ and
$\log[{2\sin(\lambda /2)}]$.
 Following the same direction,  their study is  generalized here  in
the case where the model has two seasonal fractional parameters
$d_1 $ and $d_2$ for the seasonal periods $s_1$ and $s_2$,
respectively.

{\bf Assumption 2:} $s_1$ is a multiple of  $s_2$.

{\bf Assumption 3:} Let $m=m(n)$ is a sequence satisfying
\[\left(\frac{m}{n}\right)^{\iota}\log m +\frac{1}{m}\rightarrow
0\;\;\;\mbox{ as }\; n\rightarrow \infty\] for some $\iota>0$.

{\bf Assumption 4:} \[\frac{m^{\alpha^* +0.5}}{n^{\alpha^*}}\rightarrow K \;\;\;\mbox{ as }\; n\rightarrow \infty\]
where $\alpha^*=\min (\alpha_1,\alpha_2)$.

\begin{theorem} Under assumptions 1,2 and 3, as $n\rightarrow\infty$,
\begin{description}
\item \[E(\hat{d})-d= Q^{-1}b_n(1+o(1))\]
\item \[Var(\hat{d})=m^{-1}\frac{\pi^2}{6}Q^{-1}(1+o(1)),\]
\end{description}
where
\[b_n=-2\left(\begin{array}{c}
\sum_{k=0}^{[s'/2]}b_k\delta_k (2\pi)^{\al_k}\frac{\al_k}{(\al_k+1)^2}\left(\frac{m}{n}\right)^{\alpha_k}\\\sum_{k\in I_k}b_k\delta_k (2\pi)^{\al_k}\frac{\al_k}{(\al_k+1)^2}\left(\frac{m}{n}\right)^{\alpha_k}\end{array} \right),\]
\[Q=4\left(\begin{array}{cc}
\sum_{k=0}^{[s'/2]}\delta_k & \sum_{k\in I_k}\delta_k \\
\sum_{k\in I_k}\delta_k & \sum_{k\in I_k}\delta_k
\end{array} \right),
\]
where $I_k=\{0,k\mbox{ such that $ks_2$ is a multiple of $s'$ }\}$, $\delta_k=1$ for $k=0,s'/2$ and $\delta_k=2$ otherwise. In consequence, $\hat{d}$ is consistent.
\end{theorem}

% THEOREM 2.%%%%%%%%%%%%%%%%%%%%%%%%%%%%%%%%%%%%%%%%%%%%%%%%%%%%%%%%%%%%%%%%%%%%%%%%%%%

\begin{theorem} Under assumptions 1, 2, 3 and 4, as $n\ri\infty$,
\[\sqrt{m}(\hat{d}-d)\stackrel{d}{\ri} N\left(Q^{-1}b,\frac{\pi^2}{6}Q^{-1}\right)\]
for
\[b=-2(2\pi)^{\al^*}\frac{\al^*}{(\al^*+1)^2}K\left(\begin{array}{c}
\sum_{k\in J_k}b_k\delta_k \\\sum_{k\in I_k\bigcap J_k}b_k\delta_k \end{array} \right),\]
where $J_k=\{k \mbox{ such that } \al_k=\al^*\}$.
\end{theorem}

Proofs of the above theorems are in Appendix A.

As a particular case of   Model \ref{model7},  the  statistical
properties of the SARFIMA$(0,d,0)_s$ model, with  $\nu_t \equiv
\epsilon_t$,  are now discussed.

%asymptotic bias and variance of the proposed estimator is provided
% for  model \eqref{model7} with only one single fractional seasonal parameter, that is,
% $X_t$  is a   SARFIMA$(0,d,0)_s$ model where $s$ is an integer value.

The OLS estimator of $d$  is given by
\begin{equation}
\label{model10}
     \widehat{d} = (-0.5)\frac{\sum_{k=0}^{[\frac{ s}{2}]}\sum_{j=1}^{m}(X_{1,k,j} -
     \bar{X_1})\log I_{k,j} }{\sum_{k=0}^{[\frac{s}{2}]}\sum_{j=1}^{m}(X_{1,k,j}-\overline{X_1})^2},
\end{equation}
where $X_{1,k,j}= \log \left\{2\sin((s\lambda_{k,j}/2)\right\}$. By simple algebra, the following expression is reached.
 \begin{equation}
    \widehat{d} - d \approx - \frac{1}{2S_{X_1X_1}} \sum_{k=0}^{[\frac{ s}{2}]} \sum^{m}_{j=1}
    (X_{1,k,j}-\overline{X_1})U_{k,j},
\end{equation}
where  $S_{X_1X_1} ={\sum_{k=0}^{[\frac{ s}{2}]}\sum_{j}^{m}(X_{1,k,j}-\overline{X_1})^2}$ and $j$ is defined as in (\ref{model8}).

\begin{prop} Let  $X_t$   be a   SARFIMA$(0,d,0)_s$ model  and $\widehat{d}$
is the OLS estimator of $d$ provided by (\ref{model10}). Under assumptions 1 to 4, as $n \to \infty$,

%Suppose that the bandwidth $m$ satisfies
%the assumptions $m\to \infty $, $n \to \infty $, with $m/n \to \ 0$ and $m\log(m)/n \to
%\ 0$. Then,
\begin{itemize}
\item[(a)]$$ E(\widehat{d})  \approx d.$$
%$\begin{equation}
%\end{equation}
\item[(b)] The variance of the estimator is given by
\begin{equation}
% Var(\widehat{d})\approx
% \frac{1}{4S_{xx}^{2}} \left(
% \sum_{k=0}^{[\frac{ s-1}{2}]} \sum^{m}_{j = 1} (x_{k,j}-\overline{x})^2 \var(U_{\lambda_{k,j}}) \right) \approx \frac{\pi^2}{12sm}.
Var(\widehat{d})\approx  \frac{\pi^2}{24sm}.
  \end{equation}
%In the above,
%  $\sum_{k=0}^{[\frac{ s-1}{2}]} \sum^{m}_{j = 1} (x_{k,j}-\overline{x})^2 \approx [\frac{ s}{2}]m$ (Hurvich and Beltrao, 1994).
  \item[(c)] The estimate $\widehat{d}$ satisfies
  $$ \sqrt{m}\,(\widehat{d}-d) \stackrel{d}{\ri} N\left(0,\frac{\pi^2}{24s}\right).$$
%where $"\to"$ means convergence in distribution.

\end{itemize}
\end{prop}

\begin{proof}
The above results are  particular cases of Theorem 1 and coincide with  Theorems 1 and 2 in Hurvich, Deo and Brodsky (1998). Note that the variance of the estimator suggested in Porter-Hudak (1990) and Ray (1993) is approximately $4s Var(\widehat{d})$.
\end{proof}

%\subsection{Fox-Taqqu method }
% This estimator,  denoted hereafter FT, is a
%parametric procedure due to Fox and Taqqu (1986) and adapted here
%for Gaussian seasonal long-memory processes.   This estimator is
%obtained by using all harmonic frequencies between the seasonal
%frequencies. It is calculated   by minimizing the approximate
%Gaussian log-likelihood
%\begin{equation}
%\mathcal{L}_{W}(\boldtheta )=\frac{1}{2n}\sum_{j}^{}\left\{ \ln
%f(\lambda_{j}
%)+\frac{I(\lambda_{j})}{f(\lambda_{j} )}\right\} ,
%\end{equation}
%where $f ( \lambda)$ is the spectral density, $\boldtheta$
%denotes the   vector of unknown parameters and $\sum_{j}$ is the sum over $j = 1,\dots,n-1$, excluding
%those values $\lambda_{j}$ coinciding with the seasonal frequencies.  Under some conditions, the FT estimator for non-seasonal ARFIMA models,
%is asymptotically normal and consistent, and for Gaussian process, the estimator is also
%asymptotically efficient (Giraitis and Surgailis, 1990, Fox and Taqqu, 1986 and others).  By assuming that the poles of the spectral density
%$f(\lambda_{j})$ are known, the asymptotic theory for the
%FT method can be extended to the SARFIMA model (see the discussion in Arteche
%and Robinson(2000), Section 2).

\vspace{0.3cm}
\section{ Finite sample investigation}
\vspace{0.3cm}

The finite sample performance of the estimator discussed previously is investigated in this section through  Monte Carlo
experiments for different structures of Model \ref{model7} where $\nu_t$ follows a SARMA model. To generate the models, the
procedure used is the one suggested in Hosking (1984) with i.i.d  innovations from a  N(0,1) distribution. The models are:  SARFIMA$(p,d_1,q)_{s_1}(P,d_2,Q)_{s_2}$ with $p=P=0,1$, $s_1$ = 1,4, $s_2$ = 4, 12 and the AR non-seasonal ($\phi_1$) and seasonal ($\phi_s$) parameters with values $\phi_1$=$\phi_s$= 0.0, 0.3 and 0.8. The parameters are also displayed in the tables.
The empirical investigations were based on sample size $n$ = 1080, and the  sample quantities mean, correlation  and mean squared error ($mse$) of the estimators were calculated  over 2,000 replications. The calculations were carried out by means of an Ox program in an AMD Athlon XP 1800 computer.

Since the models also involve  short-range dynamics, the regression estimators were obtained  by using different bandwidths. In
the case where the model has not AR contribution, the bandwidth $m = [\frac{n-1}{\max(s_1, s_2)}]$ was fixed. In this
context, the regression estimator ($GPH_T$) becomes a parametric procedure. For the models with short-memory dynamics, the two bandwidths $m= n^{\alpha_i}$, $\alpha_1=0.5$ and $\alpha_2=0.3$ were used, and the estimators are denoted as $ GPH_{1} $ and $ GPH_{2} $, respectively.  The bandwidth $n^{0.5}$  is here considered because this specification   has been widely used in the case of ARFIMA models with short-memory components, while the choice of   $n^{0.3}$ is based on the empirical investigation discussed below.

For a comparison purpose between semiparametric and parametric approaches, the procedure due to Fox and Taqqu (1986) (FT),  which  possesses good asymptotic
properties, is also adapted here for Gaussian seasonal long-memory processes.
This estimator is obtained by using all harmonic frequencies between the seasonal frequencies. It is calculated   by minimizing the approximate Gaussian log-likelihood
\begin{equation}
\mathcal{L}_{W}(\boldtheta )=\frac{1}{2n}\sum_{j}^{}\left\{ \ln
f(\lambda_{j}
)+\frac{I(\lambda_{j})}{f(\lambda_{j} )}\right\} ,
\end{equation}
where $f ( \lambda)$ is the spectral density, $\boldtheta$
denotes the   vector of unknown parameters and $\sum_{j}$ is the sum over $j = 1,\dots,n-1$, excluding
those values $\lambda_{j}$ coinciding with the seasonal frequencies.  Under some conditions, the FT estimator for non-seasonal ARFIMA models,
is asymptotically normal and consistent, and for Gaussian process, the estimator is also
asymptotically efficient (Giraitis and Surgailis, 1990, Fox and Taqqu, 1986 and others).  By assuming that the poles of the spectral density
$f(\lambda_{j})$ are known, the asymptotic theory for the
FT method can be extended to the SARFIMA model (see the discussion in Arteche
and Robinson(2000), Section 2). It should be noted that the focus of this paper is to estimate the seasonal fractional parameters only even though the parametric FT method also provides estimates for the AR parameters.

Table \ref{tableS4}  summarizes the results for   the SARFIMA model with $d_1=0.3$ ($s_1=4$).  The first part of this table shows the performance of the regression methods when there is no seasonal AR contribution.   For the case of  $\phi_1 =0$, $GPH_T$ has the best performance among the GPH based ones, which is an expected result  since  the method uses all non-seasonal frequencies in the regression equation and in this sense it is parametric and comparable to the FT. The effect of the bandwidth is also a motivation of this study. The reduction of the bandwidth causes an increase in the $mse$, especially when $\phi_1$= 0.   This  is  a not surprising result, since  the AR contribution is mainly concentrated at zero frequency. The absence of short-memory component allows a wider bandwidth because the bias is quite controlled. Thus, a reduction of $m$ implies a larger variance and does not reduce the bias.

In the second part of the table, the estimates were computed when the model has the AR contribution at the seasonal period $s=4$. From this, the GPH estimates are more affected by the AR component than the previous case, the bias is strongly positive and the $mse$ also increases. In this case the AR component has spectral power not only at frequency zero but also at the seasonal ones, affecting to a greater extent the estimation of $d_1$.  The small value of the bandwidth mitigates the effect of the short-memory parameters. This  is  clearer when  $\phi_4$ changes from $0.3$ to $0.8$. Hence, in this context,  the decrease of the bandwidth produces reduction on the size of the bias and the $mse$.

In general, the FT estimates outperform the GPH estimates in terms of the $mse$, which is not surprising considering the parametric nature of the FT method in a correctly specified model. However, the FT  estimates also present a significant increase of the bias when the model has AR seasonal components. The bias of the estimates  also increases substantially when there is model misspecification, as  can be seen from the examples presented in Tables \ref{MisTable1} and \ref{MisTable2}.  The misspecification problem will be discussed in the end of this section.

The following tables present the estimates when  the models  have more than one fractional parameter. Thus,  the sample correlations between the estimates were also calculated.

% TABLE 1 ARFIMA(0,0.3,0)_4, PHI_S, S =1,4

\begin{table}[!ht]
\footnotesize{
\begin{center}
\caption{Results for the seasonal ARFIMA model with $d_1=0.3$ ($s_1=4$)
and $\phi_s, \ s=1,4$,  $n=1080$.}

\begin{tabular}{c|c|c|c|c|c}  %quantidades de colunas
\hline
$\phi_s$&estimators&\multicolumn{2}{c|}{$\hat{d}_1$}&\multicolumn{2}{c}{$\hat\phi$}\\
                                               \cline{3-6}
                                &                 &     mean&mse          & mean&mse \\
\hline

                                &GPH$_T$          &0.3004 &0.0012&---   &---      \\
$\phi_1=0.0$                    &GPH$_{1}$ &0.2988 &0.0046&---   &---      \\
                                &GPH$_{2}$ &0.2984 &0.0311&---   &---      \\
                                &FT               &0.2885 &0.0009&---   &---      \\ \hline

                                &GPH$_{1}$ &0.3002 &0.0041&---   &---      \\
$\phi_1=0.3$                    &GPH$_{2}$ &0.3074 &0.0255&---   &---      \\
                                &FT               &0.2888 &0.0009&0.2980&0.0009   \\ \hline

                                &GPH$_{1}$ &0.3097 &0.0054&---   &---      \\
$\phi_1=0.8$                    &GPH$_{2}$ &0.3085 &0.0242&---   &---      \\
                                &FT               &0.2828 &0.0011&0.8011&0.0004   \\ \hline
\hline
\hline
                                &GPH$_{1}$ &0.3459 &0.0068&---   &---      \\
$\phi_4=0.3$                    &GPH$_{2}$ &0.3114 &0.0278&---   &---      \\
                                &FT               &0.1792 &0.0463&0.4144&0.0433   \\ \hline

                                &GPH$_{1}$ &0.7281 &0.1877&---   &---      \\
$\phi_4=0.8$                    &GPH$_{2}$ &0.4144 &0.0426&---   &---      \\
                                &FT               &0.2519 &0.0076&0.8141&0.0035   \\ \hline
\end{tabular}
\label{tableS4}
\end{center}}
\end{table}

 Table \ref{tableS1S4}  displays the result when the models are
 SARFIMA $(1,d_1,0)_{s_1}$ $(1,d_2,0)_{s_2}$ with $d_1$ =0.1$(s_1=1)$, $d_2 =0.3$$(s_2=4)$,  $\phi_1$ = 0.0, 0.3, 0.8 and $\phi_4$= 0.3, 0.8 whereas  Table  \ref{tableS4S12} shows  the performance of the estimates when the SARFIMA model has seasonal periods  $s_1=4$ and $s_2=12$. From Table \ref{tableS1S4} it should be noted that the contribution of  the parameter $d_1$ is mainly  at  zero frequency. Hence, in general, the semiprametric  estimators perform similarly  to the previous case that is, the estimate of $\bf{d}$ depends on the the values of the bandwidth and of the AR counterpart. The  memory parameters are estimated  simultaneously, thus there is a balance effect between the  two  estimates $\hat{d_1}$ and $\hat{d_2}$ which justifies the negative correlation values between them. In addition,   the estimates  are balanced to have the value of $\hat{d_1} +\hat{d_2}$ approximately  equal to $d_1+d_2$ which is the total memory at zero frequency. The correlations between the GPH estimates  increases with the bandwidth and  the AR coefficients. As  was expected, the FT method presents superiority performance compared with the semiparmetric approaches.

%  TABLE 2222222222222222222222222222
%\tiny{
\begin{table} [!ht]
\footnotesize{
\centering\caption{Results for models $d_1=0.1$ ($s_1=1$), $d_2=0.3$ ($s_2=4$)
and $\phi_s, \ s=1,4$, case $n=1080$.}
%\label{}
\begin{center}
%\small{
\begin{tabular}{c|c|c|c|c|c|c|c|c}  %quantidades de colunas
\hline
$\phi_s$&estimators &\multicolumn{2}{c|}{$\hat{d}_1$}&corr.&\multicolumn{2}{c|}{$\hat{d}_2$}&\multicolumn{2}{c}{$\hat\phi$}\\
            \cline{3-4}               \cline{6-9}
    &            &     mean&mse          &         &   mean&mse   & mean&mse \\
\hline
                                &GPH$_T$          &0.1043 &0.0018&$-$0.2228&0.3010&0.0013&---   &---      \\
$\phi_1=0.0$                    &GPH$_{1}$ &0.1135 &0.0290&$-$0.5144&0.2995&0.0053&---   &---      \\
                                &GPH$_{2}$ &0.0818 &0.1327&$-$0.4164&0.3121&0.0388&---   &---      \\
                                &FT               &0.1008 &0.0006&$-$0.1188&0.2868&0.0009&---   &---      \\ \hline
                                &GPH$_{1}$ &0.1166 &0.0215&$-$0.3397&0.3098&0.0045&---   &---      \\
$\phi_1=0.3$                    &GPH$_{2}$ &0.1463 &0.1553&$-$0.4927&0.3083&0.0308&---   &---      \\
                                &FT               &0.0983 &0.0050&$-$0.4216&0.2893&0.0009&0.2997&0.0060   \\ \hline
                                &GPH$_{1}$ &0.2208 &0.0388&$-$0.5415&0.3004&0.0062&---   &---      \\
$\phi_1=0.8$                    &GPH$_{2}$ &0.1257 &0.1607&$-$0.5190&0.3148&0.0375&---   &---      \\
                                &FT               &0.1069 &0.0148&$-$0.0868&0.2819&0.0014&0.7857&0.0124   \\ \hline
\hline
\hline
                                &GPH$_{1}$ &0.1074 &0.0249&$-$0.4751&0.3404&0.0083&---   &---      \\
$\phi_4=0.3$                    &GPH$_{2}$ &0.1107 &0.1299&$-$0.4233&0.3037&0.0348&---   &---      \\
                                &FT               &0.1006 &0.0006&$-$0.0780&0.1828&0.0404&0.4135&0.0393   \\ \hline
                                &GPH$_{1}$ &0.1045 &0.0285&$-$0.5011&0.7269&0.1874&---   &---      \\
$\phi_4=0.8$                    &GPH$_{2}$ &0.0445 &0.1591&$-$0.4773&0.4251&0.0549&---   &---      \\
                                &FT               &0.1073 &0.0011&$-$0.2241&0.2422&0.0079&0.8183&0.0035   \\ \hline
\end{tabular}
\label{tableS1S4}
\end{center}}
\end{table}

%%%-END TABLE s1=1 and s2=4

%Table \ref{tableS1=4S2=12} summarizes the performance of the the estimators for an SARFIMA model with two seasonal periods $d_1 =0.1$ ($s_1 =4$) and $d_2 =0.3$ ($s_2 = 12$).
Although the model in Table \ref{tableS4S12} has  fractional parameters at seasonality  periods 4 and 12,  similar conclusions of the performance of the estimates to the previous cases are observed.

% TABLE 33333333333333333333333333333 $$$$$$$$$$$$$$$$$$$$$$$$$$$$$
%\scriptsize{

\begin{table} [!ht]
\footnotesize{
\centering\caption{Results for the SARFIMA model with $d_1=0.1$ ($s_1=4$), $d_2=0.3$ ($s_2=12$),  $\phi_s$, $s=1,4,12$ and  $n=1080$.}
\begin{center}
\begin{tabular}{c|c|c|c|c|c|c|c|c}  %quantidades de colunas
\hline
$\phi_s$&estimators&\multicolumn{2}{c|}{$\hat{d}_1$}&Corr.&\multicolumn{2}{c|}{$\hat{d}_2$}&\multicolumn{2}{c}{$\hat\phi$}\\
             \cline{3-4}               \cline{6-9}
    &            &     mean& mse          &         &   mean&mse   & mean&mse \\
\hline
                                &GPH$_T$          &0.1047 &0.0018&$-$0.3194&0.3065&0.0015&---   &---      \\
$\phi_1=0.0$                    &GPH$_{1}$ &0.0994 &0.0052&$-$0.4405&0.3071&0.0021&---   &---      \\
                                &GPH$_{2}$ &0.0723 &0.0442&$-$0.5529&0.3095&0.0113&---   &---      \\
                                &FT               &0.1022 &0.0007&$-$0.2856&0.2637&0.0021&---   &---      \\ \hline
                                &GPH$_{1}$ &0.1063 &0.0061&$-$0.5230&0.3017&0.0022&---   &---      \\
$\phi_1=0.3$                    &GPH$_{2}$ &0.0797 &0.0433&$-$0.5446&0.2954&0.0119&---   &---      \\
                                &FT               &0.1012 &0.0009&$-$0.3759&0.2630&0.0023&0.2992&0.0008   \\ \hline
                                &GPH$_{1}$ &0.1468 &0.0067&$-$0.5212&0.2861&0.0030&---   &---      \\
$\phi_1=0.8$                    &GPH$_{2}$ &0.1177 &0.0374&$-$0.5767&0.2861&0.0136&---   &---      \\
                                &FT               &0.1010 &0.0008&$-$0.1308&0.2619&0.0023&0.7996&0.0004   \\
                                \hline
\hline
\hline
                                &GPH$_{1}$ &0.2043 &0.0160&$-$0.4707&0.2813&0.0020&---   &---      \\
$\phi_4=0.3$                    &GPH$_{2}$ &0.1146 &0.0338&$-$0.4223&0.3019&0.0084&---   &---      \\
                                &FT               &0.0609 &0.0319&$-$0.1445&0.2575&0.0028&0.3430&0.0296   \\ \hline
                                &GPH$_{1}$ &0.5962 &0.2528&$-$0.4792&0.2618&0.0038&---   &---      \\
$\phi_4=0.8$                    &GPH$_{2}$ &0.2634 &0.0625&$-$0.5317&0.2912&0.0130&---   &---      \\
                                &FT               &0.0754 &0.0234&   0.1184&0.2343&0.0054&0.8121&0.0203   \\

\hline
\hline
                                &GPH$_{\alpha_1}$ &0.1055 &0.0062&$-$0.5344&0.5044&0.0439&---   &---      \\
$\phi_{12}=0.3$                 &GPH$_{\alpha_2}$ &0.1255 &0.0481&$-$0.6629&0.3480&0.0161&---   &---      \\
                                &FT               &0.1010 &0.0008&$-$0.0427&0.2881&0.0031&0.3341&0.0311\\ \hline
                                &GPH$_{\alpha_1}$ &0.0863 &0.0063&$-$0.4587&1.0064&0.5010&---   &---      \\
$\phi_{12}=0.8$                 &GPH$_{\alpha_2}$ &0.1089 &0.0354&$-$0.5466&0.7553&0.2187&---   &---      \\
                                &FT               &0.1071 &0.0017&$-$0.2963&0.1859&0.0197&0.8190&0.0042   \\ \hline

                                 \hline
\hline
\end{tabular}
\label{tableS4S12}
\end{center}}
\end{table}

\pagebreak

As an additional illustrative form to observe the method's performance, the box-plots in Figures \ref{figura1}  and \ref{figura2} show  the variation of the estimates for the model in Table \ref{tableS4S12} with   $\phi_1=0.0$ and $\phi_1=0.3$, respectively.

\begin{figure}[!ht]
\begin{tabular}{cc}
\includegraphics[width=6.0cm,height=4.5cm]{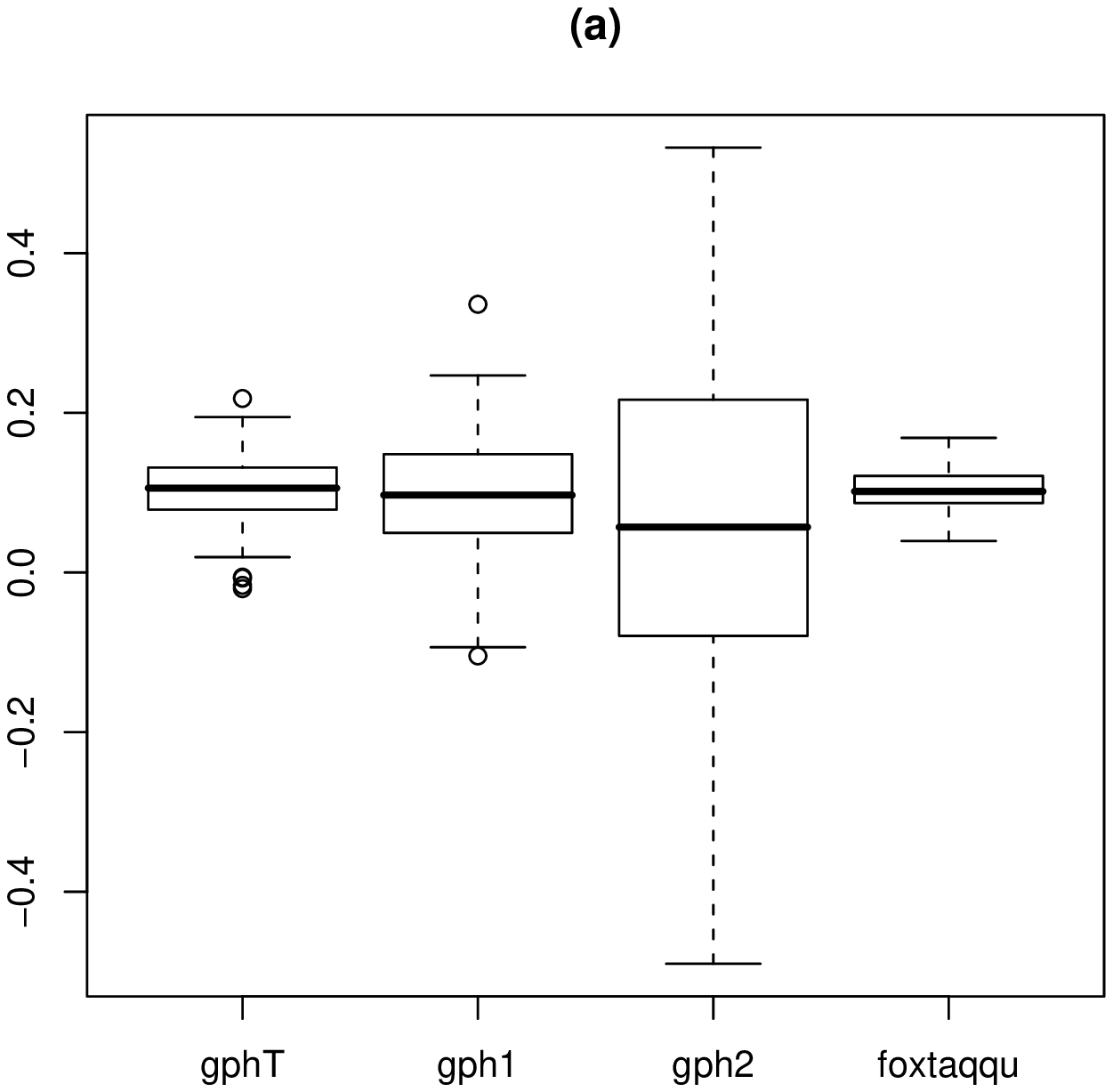}
& \includegraphics[width=6.0cm,height=4.5cm]{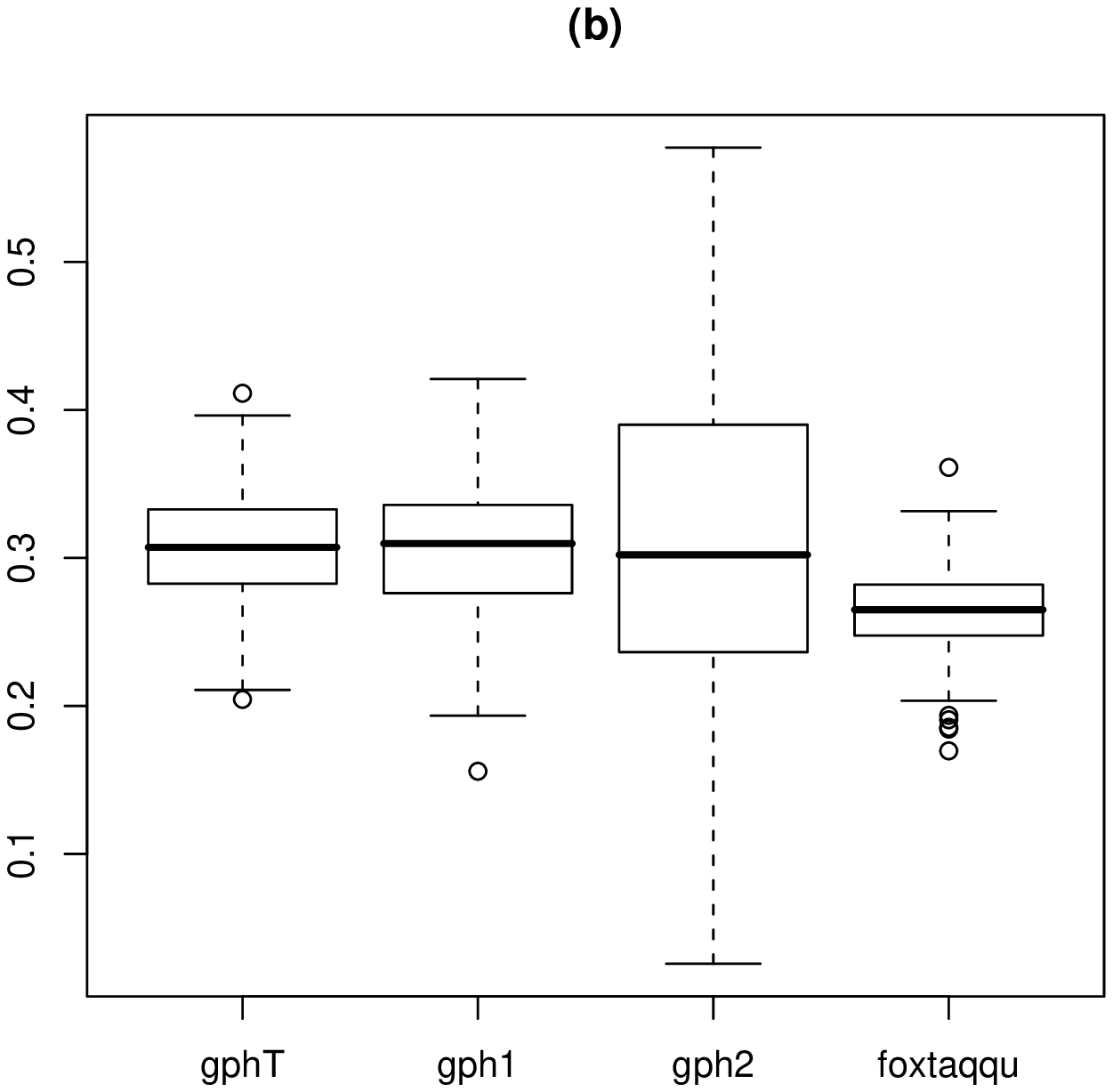}\\
%\includegraphics[width=5.5cm,height=2cm]{Latex/PeACF_m3.pdf}
%& \includegraphics[width=5.5cm,height=2cm]{Latex/PeACF_m4.pdf}
\end{tabular}
\caption{Box-plots of the estimates of $d_1$ (a) and $d_2$ (b) for the SARFIMA model with $d_1=0.1(s_1=4)$, $d_2=0.3(s_2=12)$ and $\phi=0.0$.}
\label{figura1}
\end{figure}

\begin{figure}[!ht]
\begin{tabular}{cc}
\includegraphics[width=6.0cm,height=4.5cm]{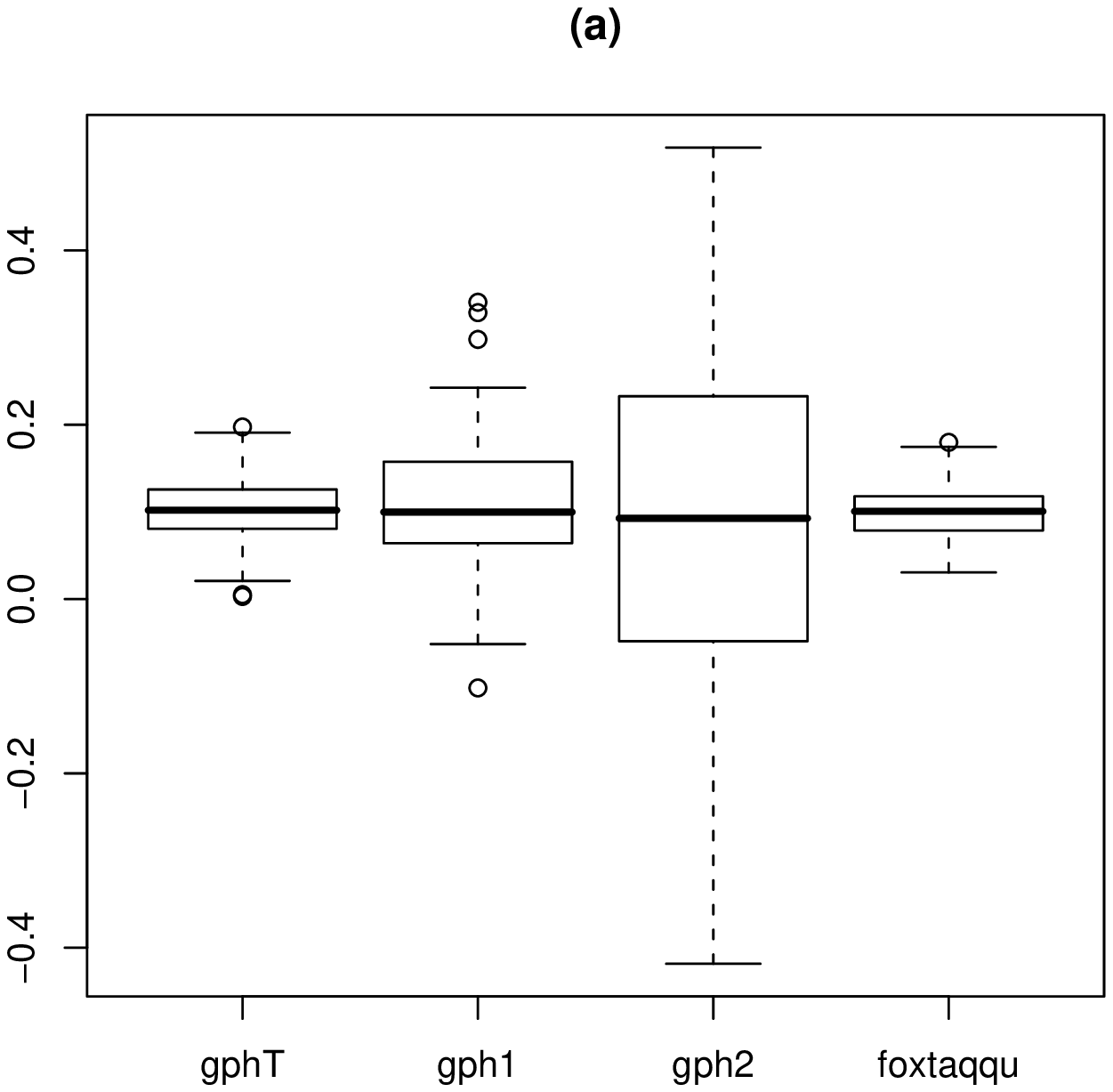}
& \includegraphics[width=6.0cm,height=4.5cm]{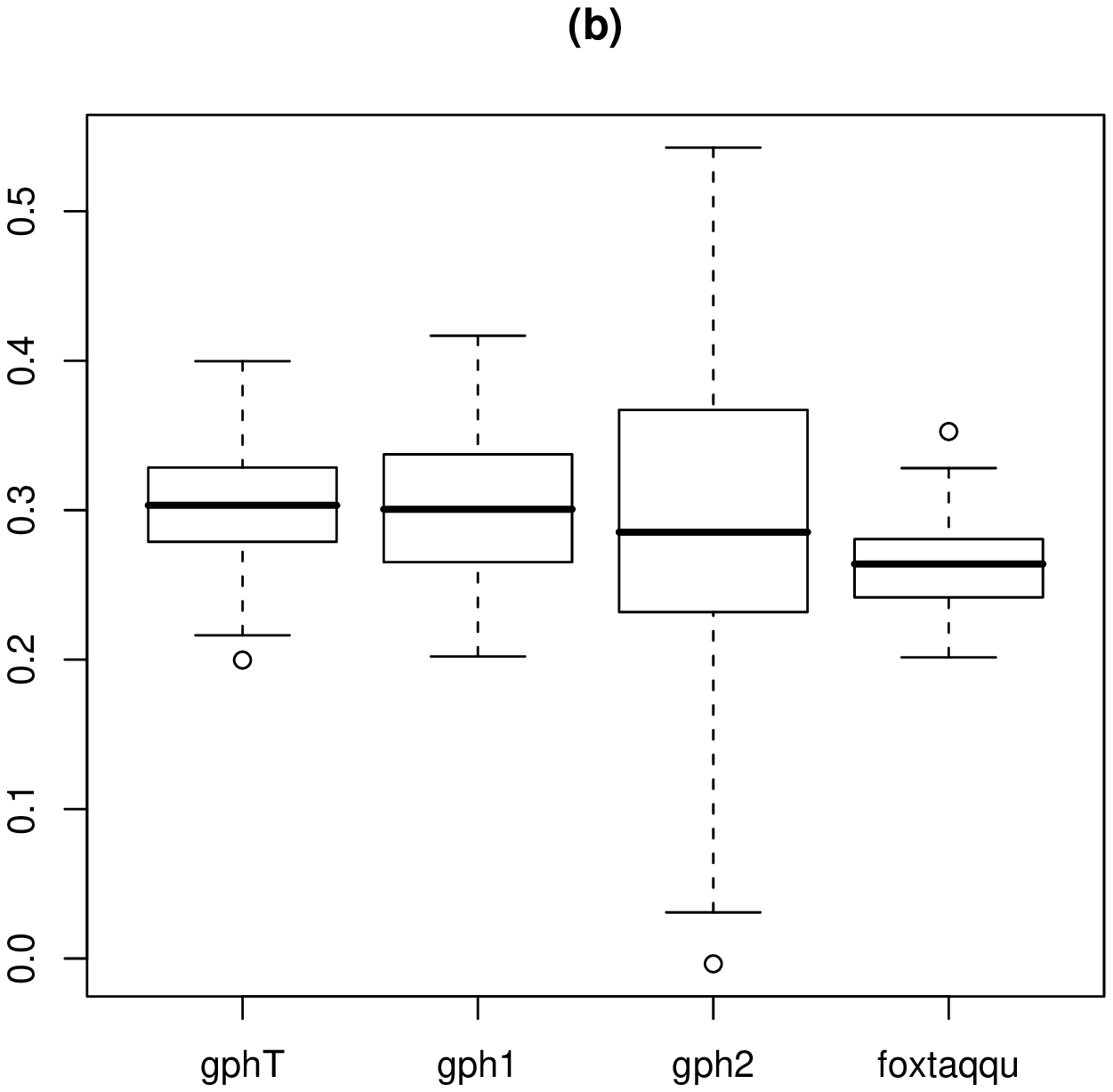}\\
%\includegraphics[width=5.5cm,height=2cm]{Latex/PeACF_m3.pdf}
%& \includegraphics[width=5.5cm,height=2cm]{Latex/PeACF_m4.pdf}
\end{tabular}
\caption{Box-plots of the estimates of  $d_1$ (a) and $d_2$ (b) for the SARFIMA model with $d_1=0.1(s_1=4)$, $d_2=0.3(s_2=12)$ and $\phi_1=0.3$.}
\label{figura2}
\end{figure}

\pagebreak
\pagebreak

 The asymptotic distribution given in Theorem 2 is also empirically investigated for the model in Table \ref{tableS4S12}  with $\phi=0$  , and the results  are depicted  in Figure \ref{figura3} which presents the empirical densities of the standardized GPH estimates of a SARFIMA model.  These figures are examples to support the claim given in Theorem 2.  The empirical densities of the estimates appear to be  fairly  close to the density of  N(0,1) distribution.

\begin{figure}[!ht]
\begin{tabular}{cc}
\includegraphics[width=6.0cm,height=4.5cm]{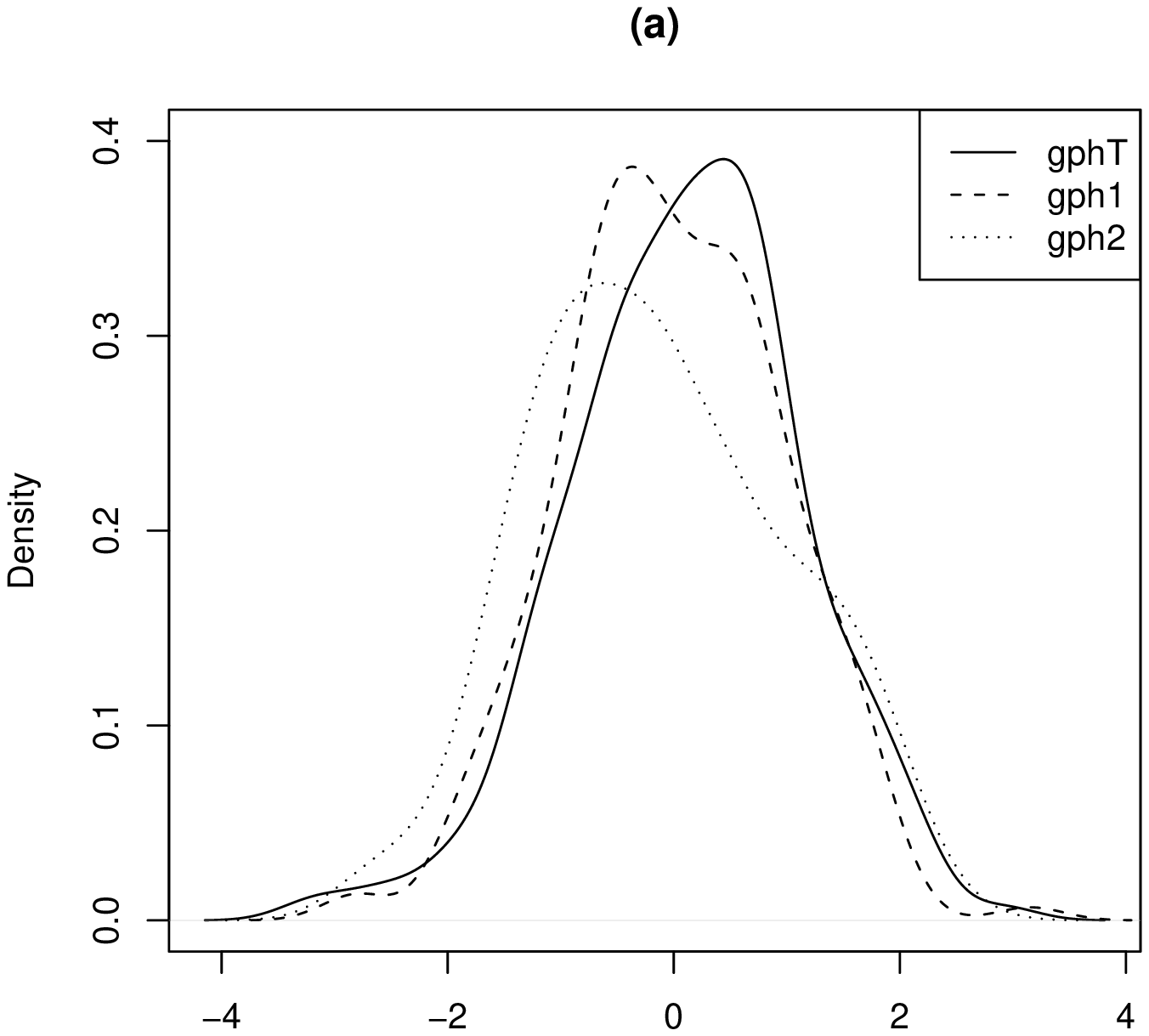}& \includegraphics[width=6.0cm,height=4.5cm]{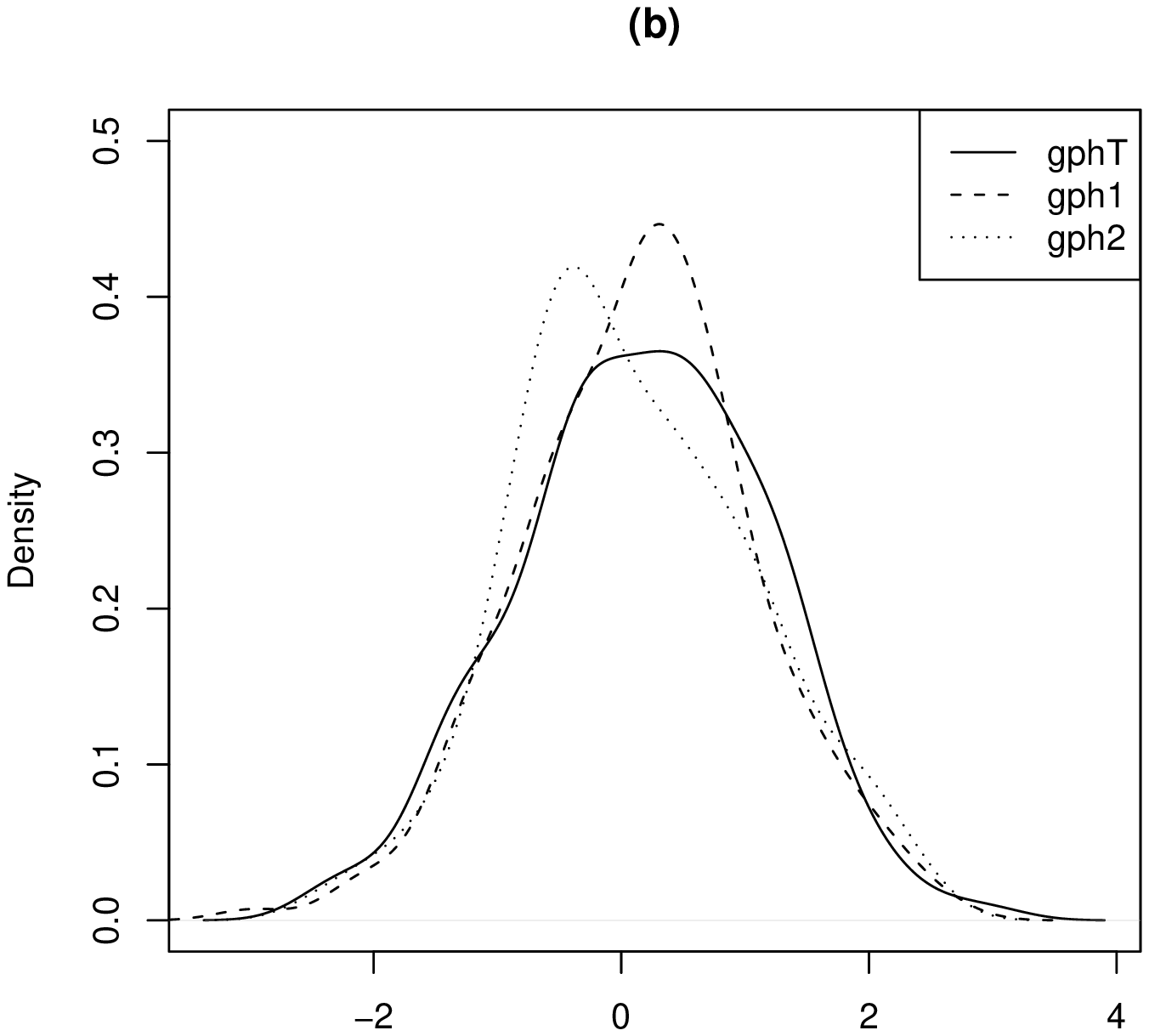}\\
%\includegraphics[width=5.5cm,height=2cm]{Latex/PeACF_m3.pdf}
%& \includegraphics[width=5.5cm,height=2cm]{Latex/PeACF_m4.pdf}
\end{tabular}
\caption{ Empirical densities of the standardized GPH estimates of   $d_1$ (a) and $d_2$ (b) for the SARFIMA model with $d_1=0.1(s_1=4)$ and  $d_2=0.3(s_2=12)$ .}
\label{figura3}
\end{figure}

\pagebreak

\textit{Model Misspecification }

As previously mentioned, the FT method was included in the study to compare the finite sample property between  semiparametric and parametric approaches. This comparison would be unfairly biased against  the semiparametric estimator proposed here  if the empirical investigation was only based  under correct model specification. Then, the next two tables deal with the estimation of the model under model misspecification of some models considered in Tables \ref{tableS1S4} and \ref{tableS4S12}. The simulated models have AR parts whereas  these short-memory parameters are omitted in the estimated models. Thus, the order misspecification is related to the non specification of the short-memory dynamics in the estimated model. Tables \ref{MisTable1} and \ref{MisTable2} give the FT estimates for the SARFIMA$(0,d_1,0)_{s_1}(0,d_2,0)_{s_2}$ models with periods $s_1=1$ and $s_2=4$ and $s_1=4$ and $s_2=12$, respectively.

In  contrast to the study presented in Tables \ref{tableS1S4} and \ref{tableS4S12}, an order misspecification, however, radically alters the performance of the parametric FT method and the semiparametric GPH here proposed tends to perform significantly better. The FT estimates  are highly biased.  It is not surprising that the significative increase of the bias and  $mse$  are closely  related to omitting  the seasonal or non-seasonal AR  value. For example, in  the first two cases of Table \ref{MisTable1}, the order misspecification is due to the non-seasonal short-dynamic part. It produces a significative positive bias of the non-seasonal memory parameter ($d_1$) whereas the estimate of $d_2$ is much less affected.  The bias and the $mse$ increase significantly when the seasonal or non-seasonal AR parameter  is close to the non-stationary region.  In the second part of Table \ref{MisTable1}, the results are on the contrary to the estimation performance of the vector $\textbf{d}$  observed in the first part of the table. Because the short-term  contributions are now at a seasonal period, the estimate of the memory parameter at $s=4$, $d_2$,  is much more affected than  $d_1$.

In Table \ref{MisTable2}, the slowly decaying autocorrelations are at period lags $s_1$ = and $s_2$ =12.   The performance of the parametric FT method is very similar to that previously considered. The biases of the seasonal memory parameters are directly related to the period  and magnitude of the AR seasonal and non-seasonal coefficients.

%Tables \ref{MisTable1} and \ref{MisTable2} report  simple examples of the impact of the model misspecification on the estimate of FT when the estimated models are the SARFIMA$(0,d_1,0)(0,d_2,0)$ models with $d_1$= 0.1 ($s_1$=1), $d_2$= 0.3 ($s_2$=4) ( Table \ref{MisTable1}) and $d_1$= 0.1 ($s_1$=4), $d_2$= 0.3 ($s_2$=12) ( Table \ref{MisTable2}). All the simulated models are displayed in the tables.

%In  \ref{MisTable1} are results when the  true models are SARFIMA$(1,d_1,0)(0,d_2,0)$, $d_1$= 0.1 ($s_1$=1), $d_2$= 0.3 ($s_2$=4), $\phi_1$ =0.3,0.8 and SARFIMA$(0,d_1,0)(1,d_2,0)$, $d_1$= 0.1 ($s_1$=1), $d_2$= 0.3 ($s_2$=4), $\phi_4$ =0.3,0.8.  \ref{MisTable2} displays the estimates when the true models are SARFIMA$(1,d_1,0)(0,d_2,0)$, $d_1$= 0.1 ($s_1$=4), $d_2$= 0.3 ($s_2$=12), $\phi_1$ =0.3,0.8 and SARFIMA$(0,d_1,0)(1,d_2,0)$, $d_1$= 0.1 ($s_1$=1), $d_2$= 0.3 ($s_2$=4), $\phi_4$ =0.3,0.8.

%%% TABLES MISPECIFICATION
%\tinny{
\begin{table} [!ht]
\centering\caption{Results of FT estimates of SARFIMA$(0,d_1,0)_{s_1}(0,d_2,0)_{s_2}$ models. The true SARFIMA models have the parameters $d_1=0.1$ ($s_1=1$), $d_2=0.3$ ($s_2=4$) and $\phi_s, \ s=1,4$, $n=1080$. }
\label{MisTable1}
\footnotesize{
\begin{center}
\begin{tabular}{c|c|c|c|c|c}  %quantidades de colunas
\hline
$\phi_s$&\multicolumn{2}{c|}{$\hat{d}_1$}&Corr.&\multicolumn{2}{c}{$\hat{d}_2$} \\
         \cline{2-3}               \cline{5-6}
     &     mean&mse         & &  mean&mse          \\
\hline
$\phi_1=0.3$&0.3270 &0.0523&$-$0.1504&0.2404&0.0043   \\ \hline
$\phi_1=0.8$&0.8503 &0.5638&$-$0.0553&0.2193&0.0073   \\ \hline
\hline
\hline
$\phi_4=0.3$ &0.0990 &0.0008&$-$0.1641&0.5193&0.0488    \\ \hline
$\phi_4=0.8$&0.1063 &0.0014&$-$0.4924&1.0398&0.5485   \\ \hline
\end{tabular}
\end{center}}
\end{table}

\begin{table} [!ht]
\centering\caption{Results of FT estimates of the SARFIMA$(0,d_1,0)_{s_1} (0,d_2,0)_{s_2}$ models. The true SARFIMA models have the parameters $d_1=0.1$ ($s_1=4$), $d_2=0.3$ ($s_2=12$) and $\phi_s, \ s=1,4\ and \ 12$, $n=1080$. }
\label{MisTable2}
\footnotesize{
\begin{center}
\begin{tabular}{c|c|c|c|c|c}  %quantidades de colunas
\hline
$\phi_s$&\multicolumn{2}{c|}{$\hat{d}_1$}&Corr.&\multicolumn{2}{c}{$\hat{d}_2$} \\
         \cline{2-3}               \cline{5-6}
              &     mean&mse         & &  mean&mse          \\
\hline
$\phi_1=0.3$ &0.1079 &0.0010&$-$0.1582&0.2600&0.0026   \\ \hline
$\phi_1=0.8$&0.4344 &0.1146&$-$0.0411&0.1718&0.0184   \\ \hline
\hline
\hline
$\phi_4=0.3$&0.3355 &0.0562&$-$0.2795&0.1968&0.0118    \\ \hline
$\phi_4=0.8$ &0.8647 &0.5855&$-$0.1354&0.1748&0.0168   \\ \hline
\hline
\hline
$\phi_{12}=0.3$&0.1059 &0.0007&$-$0.0739&0.5068&0.0435      \\ \hline
$\phi_{12}=0.8$&0.1158 &0.0018&$-$0.5797&1.0035&0.4962   \\ \hline
\end{tabular}
\end{center}}
\end{table}

\pagebreak

\vspace{0.3cm}
\section{ Examples of Application}
\vspace{0.3cm}

This section illustrates the usefulness  of the SARFIMA model and the semiparmetric fractional estimator using three examples. The first two examples are artificial series and the third example consists of the analysis of  daily average PM$_{10}$ concentrations.

%Examples of single artificial and real series are considered in this section  to illustrate . Single artificial series from the  SARFIMA model were generated and estimated, and  these examples are discussed in the next subsection. As a real example of application, the analysis of    in the SARFIMA modeling framework is the motivation of the subsequent subsection.

\subsection{Artificial data }

Samples, with sizes $n=1080$, from  models SARFIMA$(0,d_1,0)_{s_1} (0,d_2,0)_{s_2}$ (Model I) and SARFIMA$(1,d_1,0)_{s_1} (0,d_2,0)_{s_2}$ (Model II), with $d_1 =0.1 (s_1 = 4)$, $d_2=0.3 (s_2 = 12)$ and  $\phi_4 = 0.3$, were simulated according to the data generating process described in the previous section. The non-zero AR parameter is the only  factor that differentiates  the two models. So, the influence of a short-memory parameter in the estimation of the fractional memory parameters of a single series is the main purpose of the analysis of these artificial data sets.   The sample autocorrelation functions (ACF) are in Figures \ref{figuraMoldelIandII}(a) and (b) for Models I and   II, respectively, and the estimates of the models are in Table \ref{EstMoldelsIandII}.

%%% FIGURE ACF artificial Data  SAMPLE ACFs of MODEL I AND II
%\pagebreak

\begin{figure}[!ht]
\begin{tabular}{cc}
\includegraphics[width=6.0cm,height=4.5cm]{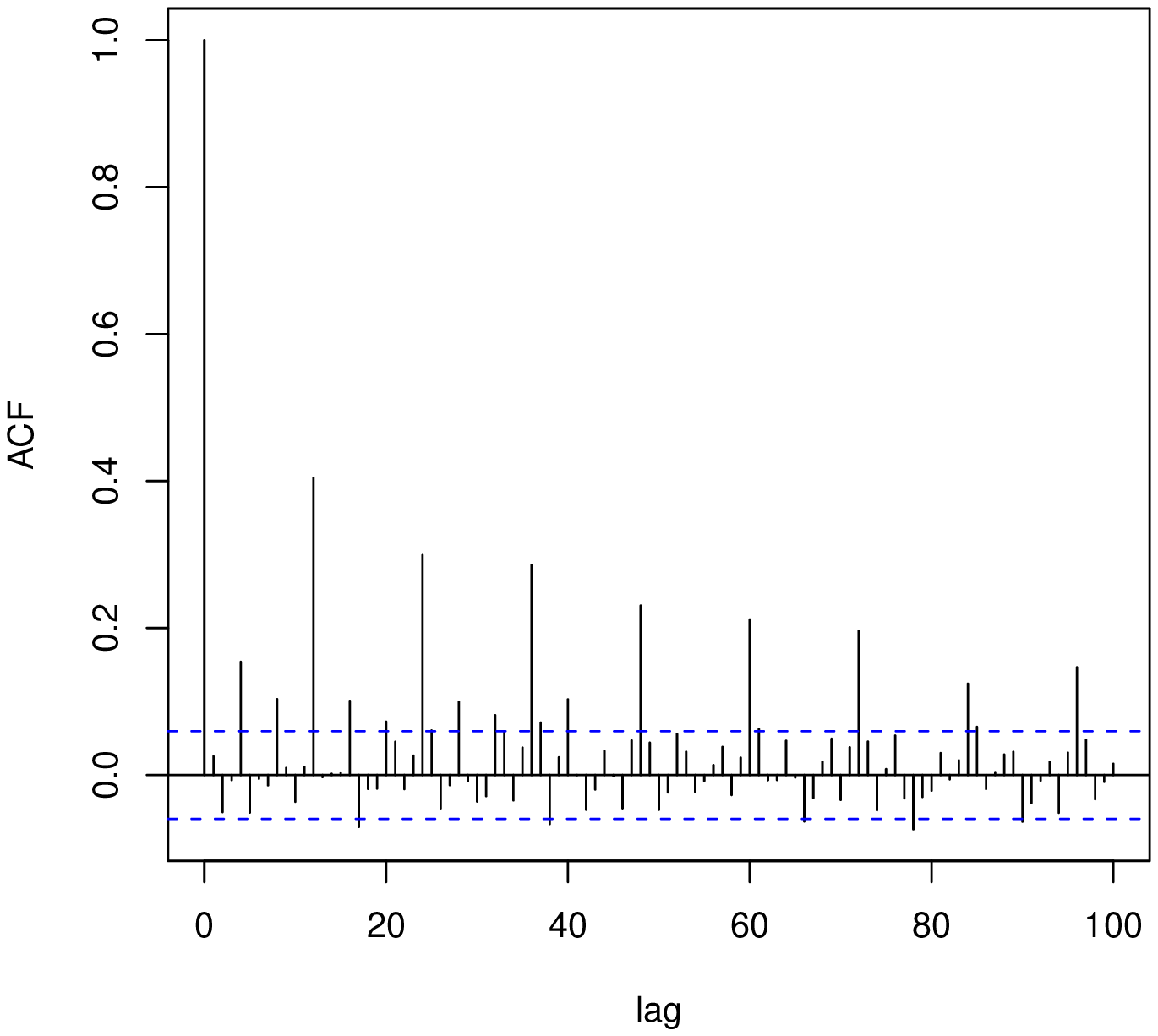}
& \includegraphics[width=6.0cm,height=4.5cm]{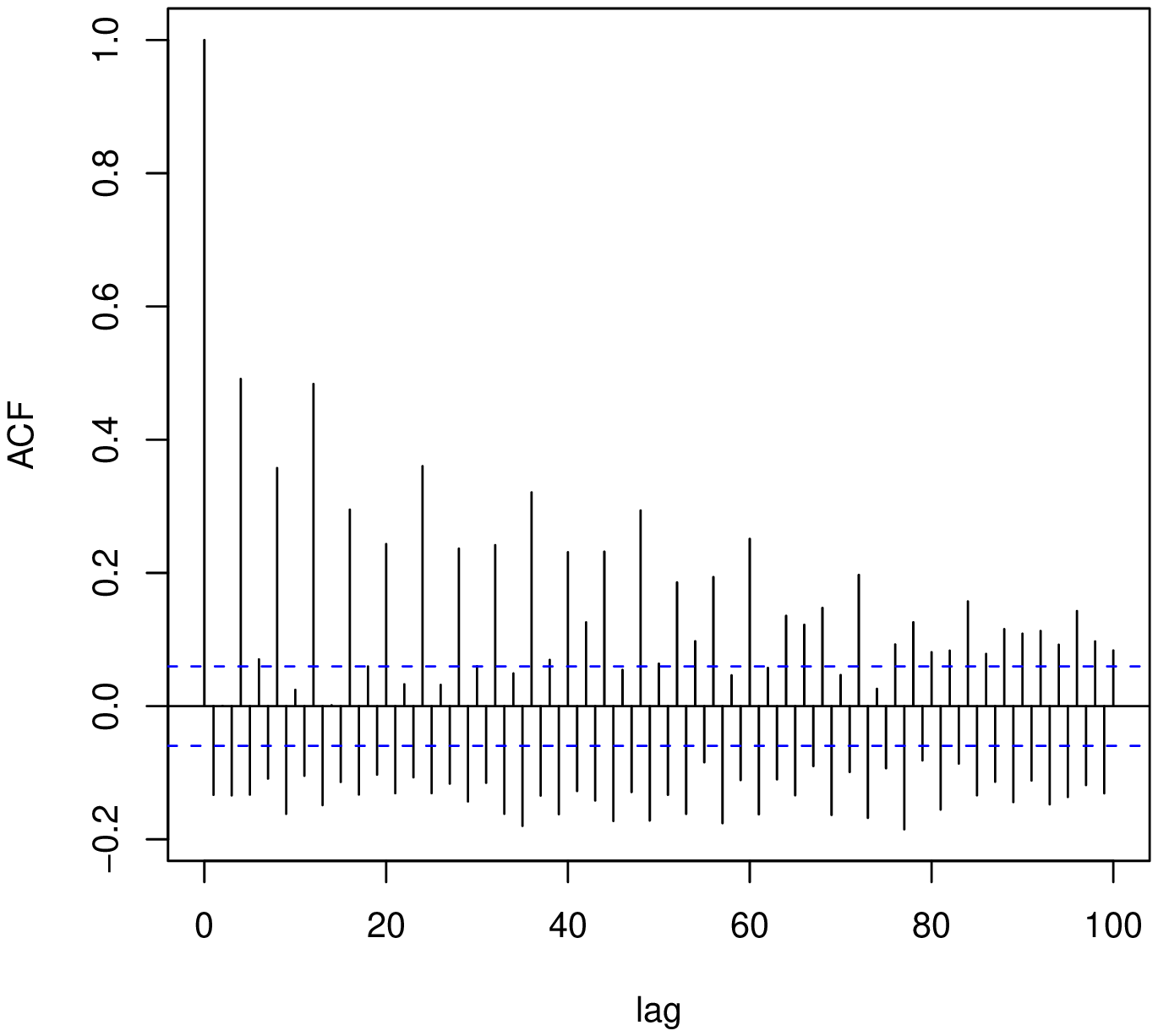}\\
%\includegraphics[width=5.5cm,height=2cm]{Latex/PeACF_m3.pdf}
%& \includegraphics[width=5.5cm,height=2cm]{Latex/PeACF_m4.pdf}
\end{tabular}
\caption{ (a) Sample ACF of Model I (b) Sample ACF of Model II.}
\label{figuraMoldelIandII}
\end{figure}

%\begin{figure}[!ht]
%	\centering
%%	  \includegraphics[width= 1\textwidth]{series.pdf}
%\includegraphics[width= 1\textwidth]{ACFsimulada_s1=4d1=01s2=12d2=03.pdf}
%	\caption{ Sample ACF of SARFIMA$(0,d_1,0)_{s_1} (0,d_2,0)_{s_2}$ model with $d_1 =0.1 (s_1 = 4)$, $d_2=0.3 (s_2 = 12)$ .}
%	\label{ACFartModelI}
%\end{figure}
%
%% FIGURE 2
%\begin{figure}[!ht]
%	\centering
%%		\includegraphics[width=1.00\textwidth]{acf.pdf}
%  \includegraphics[width=1.00\textwidth]{ACFsimulada_s1=4d1=01s2=12d2=03sphi=4phi=03.pdf}
%	\caption{ Sample ACF of SARFIMA$(1,d_1,0)_{s_1} (0,d_2,0)_{s_2}$ model with $d_1 =0.1 (s_1 = 4)$, $d_2=0.3 (s_2 = 12)$ and $\phi_4 = 0.3$.}
%	\label{ACFartModelII}
%\end{figure}

%\pagebreak

Model I does not have the AR part.  As  expected,  the ACF only has    significant  spike at lags which are multiples of 4 and 12, and they appear to have a slow decay pattern (see Figure \ref{figuraMoldelIandII}a). The seasonal autocorrelations related to the fractional parameter  $d_1$ $(s_1 = 4)$ only become  insignificant  after lag 40. So,  if the data had been analyzed with no prior information, inspection of the sample autocorrelation would indicate that the series has  seasonal periods $s= 4$ and 12 with possible long-memory structure. As can be seem in Figure  \ref{figuraMoldelIandII}b, even though  both models have the same seasonal long-memory parameters, an introduction of a positive AR coefficient may produce a significant impact on the correlation structure. The ACF of Model II (Figure \ref{figuraMoldelIandII}b) also shows  slow decaying behavior at seasonal periods, however, with a stronger correlation structure  at and  between the seasonal periods than  Model I and, in general, the patterns of this ACF are more complicated.

Apart from the usual steps for the identification of  a single long-memory time series, the estimation of the fractional parameter based on different choice of the bandwidth may be an additional tool  when using semiparametric approaches to estimate the fractional parameters. This is exemplified in Table \ref{EstMoldelsIandII}. The estimates of the memory parameters were calculated using different bandwidths $m=n^\alpha$, $\alpha =0.35 (0.3),.., max$, where $m$ satisfies the condition previously stated. For each $\alpha$, the corresponding total number of frequencies used in the regression is given in parenthesis. Table \ref{EstMoldelsIandII} also displays the estimates and their empirical  $mse$-the square of the bias plus the OLS variance. The smallest value of  $mse$ is given in bold.

\begin{table} [!ht]
%\tine{
\centering\caption{Estimated parameters of Models I and II. }
\label{EstMoldelsIandII}
\begin{center}
\small{
\begin{tabular}{c|c|c|c|c|c|c|c|c}  %quantidades de colunas
\hline
$\phi_s$&Estim.&\multicolumn{7}{c}{$\alpha$}
               \\\cline{3-9}
  & &  0.35  &0.38   &0.41&0.44&0.47& 0.50& 0.54 ( max.) \\
               & &  (132)  &(174)   &(210)& (258)&(318)& (390)& (534) \\
\hline
                       &$\hat d_1$        &0.1483 & 0.1407 & 0.1185 & 0.1208 & 0.1084 & 0.1112 & 0.1241     \\
$\phi_4=0.0$           &$\hat d_2$        &0.4117 & 0.3700 & 0.3518 & 0.3284 & 0.3188 & 0.3021 & 0.3133     \\
%                       &var ($\hat d_1$)   &0.0234 & 0.0172 & 0.0132 & 0.0098 & 0.0075 & 0.0057 & 0.0017     \\
%$\phi_4=0.0$           &var ($\hat d_2$)    &0.0080 & 0.0059 & 0.0046 & 0.0035 & 0.0028 & 0.0022 & 0.0015     \\
%                       &$m$             &132    & 174    & 210    & 258    & 318    & 390    & 534        \\
                       &mse($\hat d_1$)    &0.0257 & 0.0189 & 0.0135 & 0.0102 & 0.0076 & 0.0058 & \bf{0.0023}     \\
                       &mse($\hat d_2$)   &0.0205 & 0.0108 & 0.0073 & 0.0043 & 0.0032 & 0.0022 & \bf{ 0.0017}     \\ \hline

\hline

                       &$\hat d_1$        &0.1492 & 0.1396 & 0.1914 & 0.0997 & 0.0887 & 0.1449 & 0.2535     \\
$\phi_4=0.3$           &$\hat d_2$        &0.3138 & 0.3715 & 0.3139 & 0.3187 & 0.3199 & 0.3144 & 0.2758     \\
                       %&var( $\hat d_1$)    &0.0269 & 0.0213 & 0.0157 & 0.0115 & 0.0083 & 0.0066 & 0.0019     \\
%$\phi_4=0.3$           &var( $\hat d_2$)    &0.0092 & 0.0074 & 0.0055 & 0.0041 & 0.0031 & 0.0026 & 0.0017     \\
%                       &$m$             &132    & 174    & 210    & 258    & 318    & 390    & 534        \\
                       &mse( $\hat d_1$)    &0.0293 & 0.0229 & 0.0241 & 0.0115 & \bf{0.0084} & 0.0086 & 0.0255     \\
                       &mse ($\hat d_2$)    &0.0094 & 0.0125 & 0.0057 & 0.0044 & 0.0035 & 0.0028 & \bf{0.0023}     \\ \hline

\end{tabular}
}
\end{center}
\end{table}

% FIGURES OF THE MSE OF BOTH MODELS
\begin{figure}[!ht]
\begin{tabular}{cc}
\includegraphics[width=6.0cm,height=4.5cm]{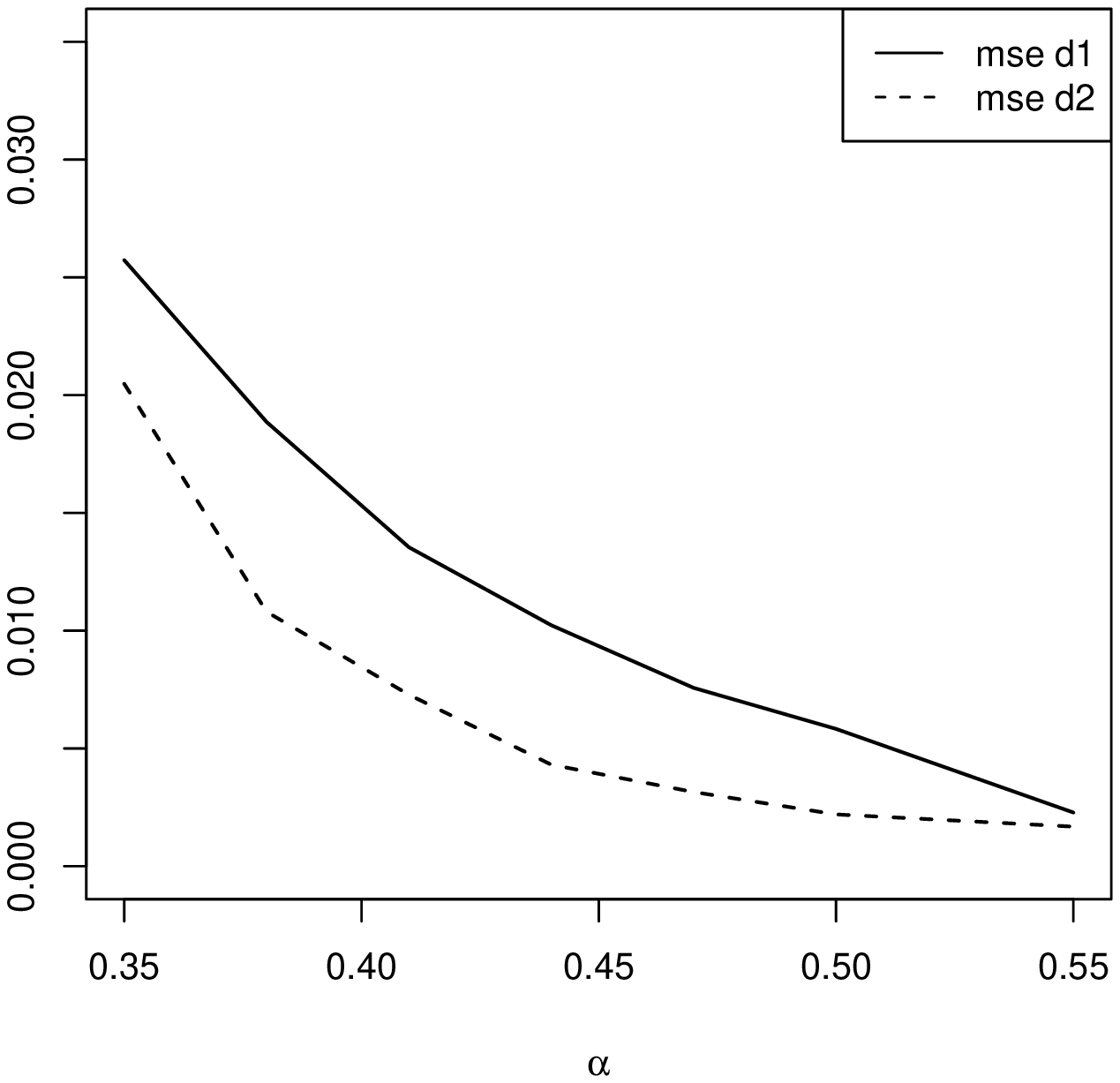}
& \includegraphics[width=6.0cm,height=4.5cm]{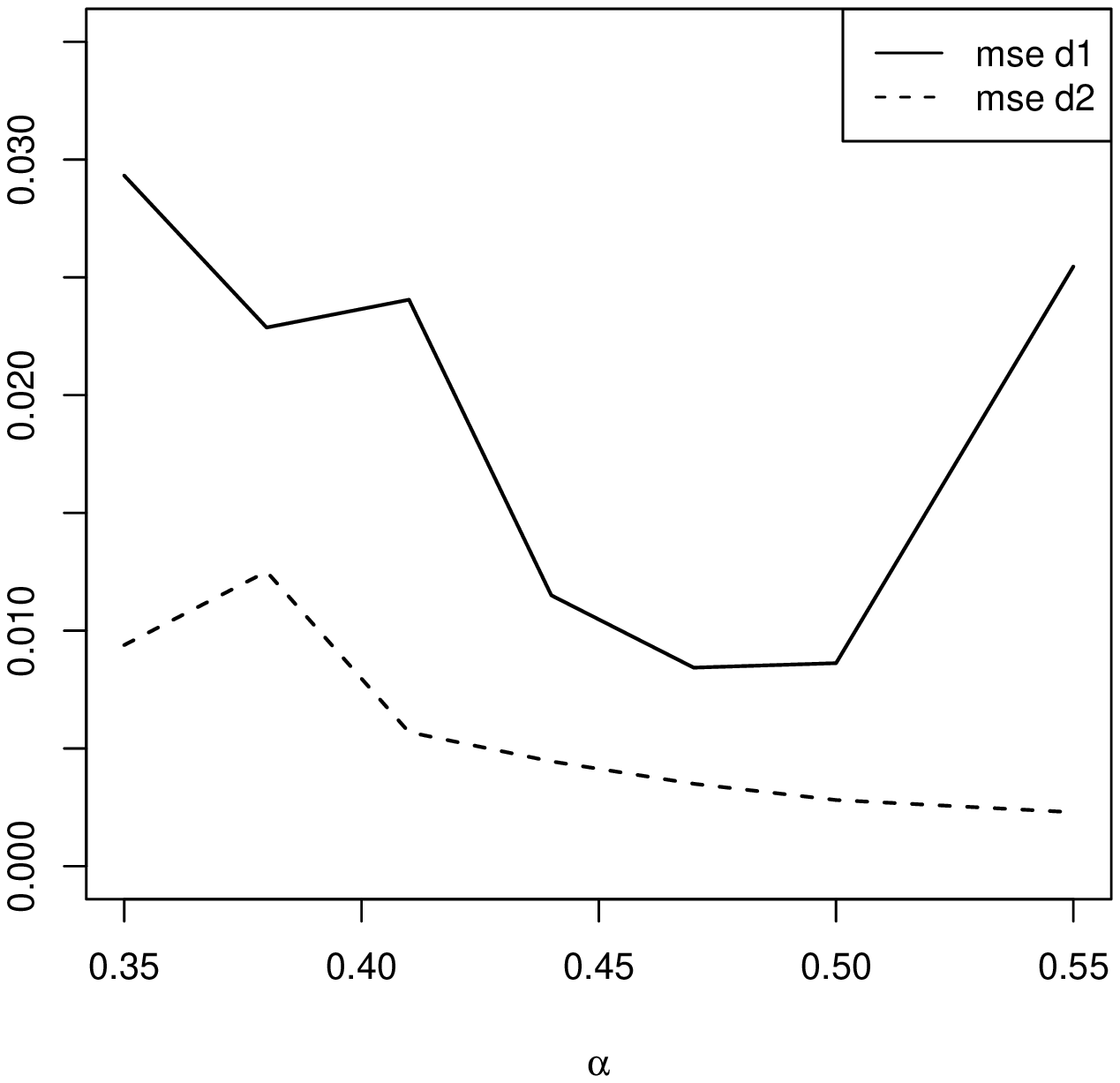}\\
%\includegraphics[width=5.5cm,height=2cm]{Latex/PeACF_m3.pdf}
%& \includegraphics[width=5.5cm,height=2cm]{Latex/PeACF_m4.pdf}
\end{tabular}
\caption{ (a) $mse$'s estimates of the fractional parameters  of Model I (b)  $mse$'s estimates  of the fractional parameters  of Model II.}
\label{figuraMSEmodelIandII}
\end{figure}

 As was expected, the presence of seasonal and nonseasonal short-memory components may bias the estimates of the fractional parameters, which is in accordance with the simulation results presented in Section 4. The bias may be reduced by an appropriate choice of the bandwidth, however, it is not easy task  in a real practical application. So, looking at the behavior of the estimates across different bandwidth values may, at least, indicate  the unfavorable estimates. In the samples here considered, it can be observed that when there is no short-memory part (Model I), the estimates are, in general, very stable across the bandwidth. The reduction of the variance and, also, the $mse$ is obtained by increasing $m$.

  In Model II   the seasonal period of the short-memory is  $s=4$, so  the corresponding fractional estimate is more affected, whereas the fractional estimate of period $s=12$ remains more stable in a wide range of bandwidths. Now, the smallest $mse$ of the estimates are not at the same bandwidths. Since the estimate of $d_1$ is more affected with the AR part, its estimate has the smallest $mse$ with a smaller number of regression than the estimate of $d_2$. To have a better understanding of the behavior of the $mse$ across the size of the bandwidth, these values are displayed in Figure \ref{figuraMSEmodelIandII}.
  So, from this simple example it can be seen that the bias of each fractional estimate may be substantially affected, if at the same seasonal period, there is a short-memory component. As an example of a stronger correlation AR seasonal structure,  a SARFIMA model with   $\phi_4 =0.8$ was also considered, but it is not presented here to save space. As  expected, the smallest $mse$ of both seasonal fractional estimates were achieved for bandwidths smaller than the case where  $\phi_4 =0.3$.

  To conclude,  in Table \ref{EstMoldelsIandII} the estimates based on  the smallest $mse$ were used to estimate the fractional parameters of the artificial series. The model adequacies for the adjusted models  were carried out, for example,  the residual analysis. These evidenced  that the estimated models fit  the series well, that is, no anomaly of the residuals were found (residual analysis of the artificial data are available upon-request).

  In the same direction of the above exercises, other single series were also analyzed with different short-memory parameters and seasonal periods, however, in general, the estimates presented similar patterns of those here presented. These are available upon-request.

\subsection{Daily average PM$_{10}$ concentration}

The daily average Particulate Matter (PM$_{10}$) concentration  is   expressed in $\mu$g/m$^3$ and it was  observed in the
Metropolitan Region of Greater Vitória (RGV) in Brazil.  RGV  is comprised of five  cities with a population of approximately 1.7 million inhabitants in an area of 1,437 $km^2$. The region is situated on the South Atlantic coast of Brazil (latitude 20°19S, longitude 40°20W) and has a tropical humid climate, with average temperatures ranging from $24^{o}C$ to $30^{o}C$. The rainfall is fairly distributed throughout the entire year (average precipitation of 98.3 \textit{mm} per month during the period of study), but with drier periods from June to August (average precipitation of 60.8 \textit{mm} per month) and more heavier precipitation from October to January (average precipitation of 158.3 \textit{mm} per month).

The raw series has a sample  size of  2037 observations, measured from the   1st of January 2001 to 2nd of August 2006,  and it is shown graphically in   Figure \ref{PM10series}. The sample autocorrelation (ACF)   and partial autocorrelation (PACF)
functions   are shown in Figures \ref{sampleACFPACFPM10} (a) and (b), respectively. From these plots   a strong seasonal component in the series is evident, which was an expected   property  due to the characteristic of such a physical phenomena. It is also observed that the seasonality behavior has period $s=7$, which is also an expected data behavior since the series was observed daily.

An interesting feature observed from the sample ACF  is the   slow decay of the correlations in the first lags, in the lags multiple of 7 and in the lags  between the seasonal periods. The ACF plot strongly indicates that the process has fractional  memory parameters in the lung-run and in the seasonal periods. This empirical evidence indicates the use of a particular case of the SARFIMA model defined previously (Model \eqref{model7}) with $s_1=1$, $s_2 =7$. The modeling strategy   follows the  same steps suggested in Hosking (1981) and investigated empirically by Reisen (1994), Reisen \& Lopes (1999) among others. Firstly, the fractional  parameters   are  estimated by using the GPH semiparametric tool described in the previous section.  This was carried out by using different sizes of  bandwidth $m$. Secondly, the truncated filter
$(1-B)^{\hat{d_1}}(1-B^s)^{\hat{d_2}}$ is used to filter the observation and obtain a new series which approximately follows an ARMA model. This new series
is used to achieve the complete short-memory model structure.

% FIGURE PM10%%%%%%%%%%%%%%%%%%%%%%%%%%%%%%%%%%%%%%%%%%%%%%%%%%%%%%%%%%%%

% FIGURES OF THE PM10
\begin{figure}[!ht]
\centering
\begin{tabular}{cc}
\includegraphics[width=10.0cm,height=7.0cm]{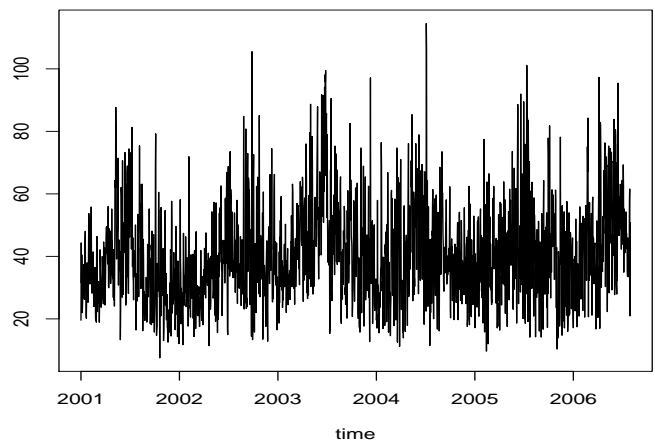}
%& \includegraphics[width=6.0cm,height=4.5cm]{MSEsimulada_s1=4d1=01s2=12d2=03sphi=4phi=03}\\
%%\includegraphics[width=5.5cm,height=2cm]{Latex/PeACF_m3.pdf}
%%& \includegraphics[width=5.5cm,height=2cm]{Latex/PeACF_m4.pdf}
\end{tabular}
\caption{ $PM_{10}$ series.}
\label{PM10series}
\end{figure}

% FIGURES OF THE ACF and PACF PM10
\begin{figure}[!ht]
\begin{tabular}{cc}
\includegraphics[width=6.0cm,height=4.5cm]{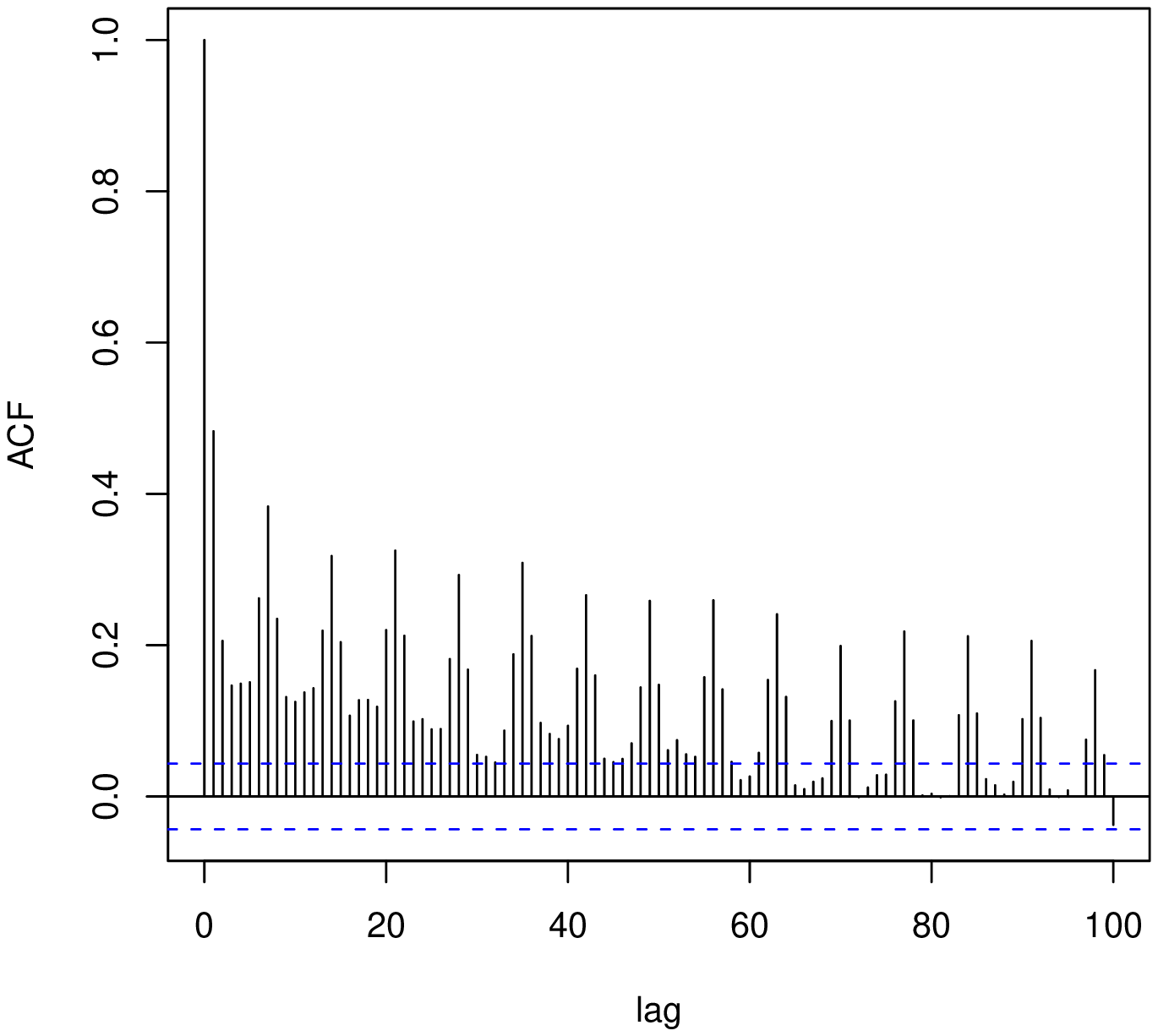}
& \includegraphics[width=6.0cm,height=4.5cm]{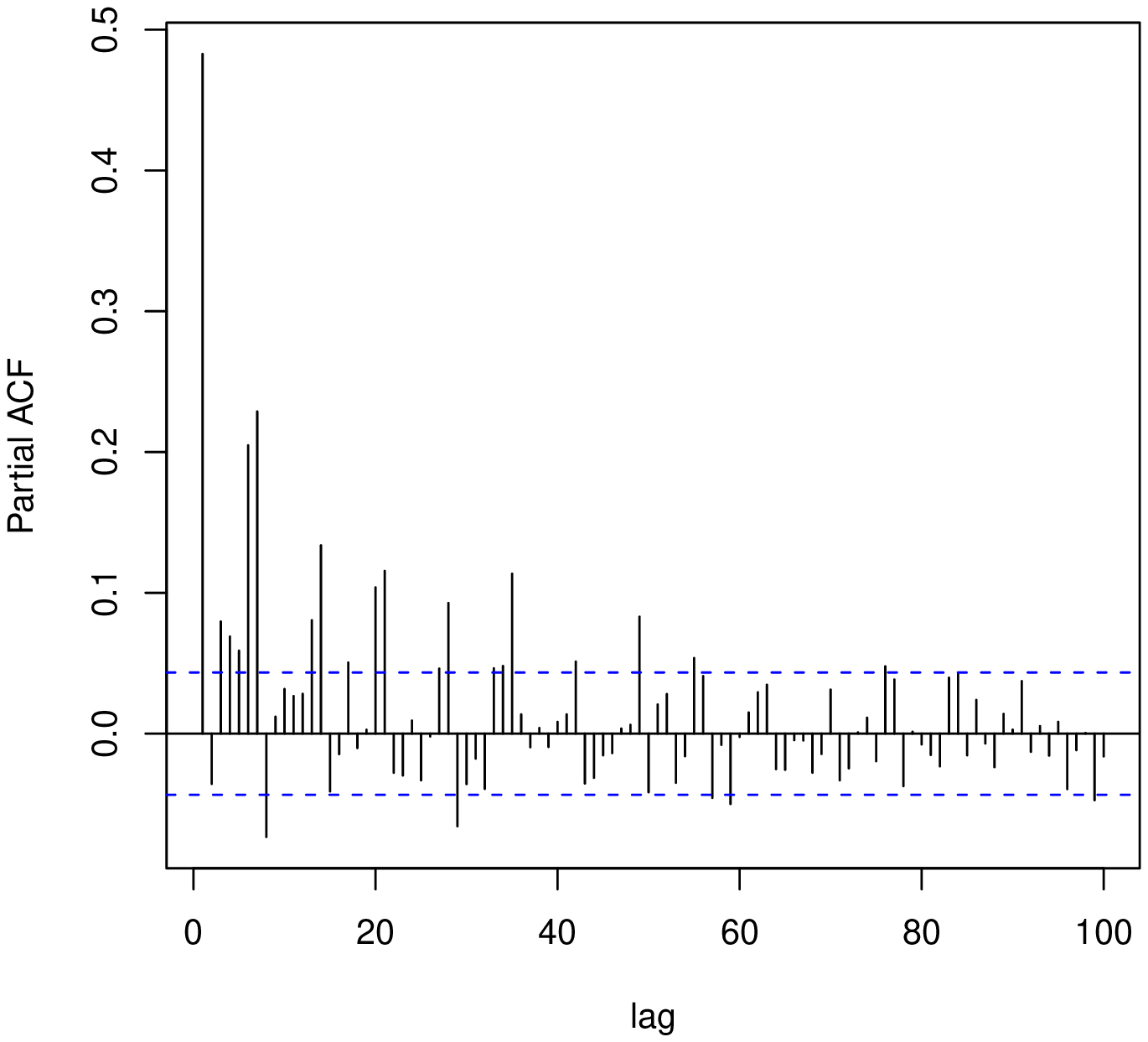}\\
%\includegraphics[width=5.5cm,height=2cm]{Latex/PeACF_m3.pdf}
%& \includegraphics[width=5.5cm,height=2cm]{Latex/PeACF_m4.pdf}
\end{tabular}
\caption{ (a) Sample ACF of $PM_{10}$ (b)  Sample PACF of $PM_{10}$.}
\label{sampleACFPACFPM10}
\end{figure}

%\begin{figure}[!ht]
%	\centering
%%	  \includegraphics[width= 1\textwidth]{series.pdf}
%\includegraphics[width= 1\textwidth]{plotseries2.pdf}
%	\caption{ PM$_{10}$ data.}
%	\label{y}
%\end{figure}
%
%% FIGURE 2
%\begin{figure}[!ht]
%	\centering
%%		\includegraphics[width=1.00\textwidth]{acf.pdf}
%  \includegraphics[width=1.00\textwidth]{acf_original.pdf}
%	\caption{ Sample autocorrelation (left) and partial
%          autocorrelation functions (right).}
%	\label{fac-y}
%\end{figure}

%\subsection{Adjusted models}

Table \ref{Estimatesd1d2PM10} displays the results of the memory estimates
obtained from different values of the  bandwidth $m=n^\alpha$,
$0<\alpha<1$ and the values in parenthesis are the corresponding   number of the frequency used in the regression.  From this table, it can be seen that the values of the estimates of $d_1$ and $d_2$ are stable for $ 0.52 <\alpha< 0.65(max.)$, and  they are in the range $0< d_1,d_2 <0.5$. The estimated standard errors of $\hat{d_2}$ are relatively small and two-sided confidence intervals for $d_1$ and $d_2$ are correspondingly tight. Therefore, for  $ \alpha > 0.52$ the null hypotheses that $H_0: d_1=0$ and $H_0: d_2=0$ are rejected. Also, for all values of the bandwidths given in the table,  F test was performed for the  null hypothesis  $H_0: \textbf{d}=\textbf{0}$, and it indicated that at least one fractional parameter is different from zero.  The stable value of the estimate of $d_1$  in the the range $ 0.52 <\alpha< 0.65(max.)$ gives an empirical evidence  that  if there is any non-seasonal short-memory part in the model, the parameter is not large enough to make a significant contribution in the regression estimators. A similar conclusion is also observed  in the case of the seasonal fractional estimate $\hat{d_2}$.
%, that is,  even for large value of $\alpha$, the estimates of $d_2$ do not present  noticeable  change.
Therefore,  $\alpha=0.54$ was chosen to estimate the memory parameters. The vector  $\hat{\textbf{d}}=(0.1918,\ 0.1798)$ shows that data presents the  stationarity,  invertibility and long-memory properties. The choice of these estimates was also confirmed by the Akaike Criterion (AIC), which gave the smallest value for $\alpha=0.54$.

%%\footnotesize{
%%\tine
%\begin{table} [!ht]
%\small{
%%\tine{
%\centering\caption{Estimates of the fractional parameters of the $PM_{10}$ }
%\label{Estimatesd1d2PM10}
%\begin{center}
%\begin{tabular}{c|c|c|c|c|c|c|c|c|c|c}  %quantidades de colunas
%\hline
%Est.            &\multicolumn{10}{c}{$\alpha$}
%                  & \cline{2-11}
%                  &  0.44 & 0.46   & 0.48   & 0.50   &  0.52  &  0.54  &  0.56   &  0.58  &  0.6   & 0.65(max.)   \\
%\hline
%
%$\hat d_1$          &0.1633 & 0.1452 & 0.1112 & 0.0970 & 0.1004 & 0.1918 & 0.1787 & 0.2071 & 0.1645 & 0.2534   \\
%$\hat d_2$          &0.2927 & 0.2768 & 0.2585 & 0.2223 & 0.1954 & 0.1798 & 0.1575 & 0.1443 & 0.1548 & 0.0806   \\
%var $\hat d_1$      &0.0252 & 0.0221 & 0.0174 & 0.0134 & 0.0110 & 0.0090 & 0.0072 & 0.0060 & 0.0049 & 0.0014   \\
%var $\hat d_2$      &0.0037 & 0.0032 & 0.0026 & 0.0020 & 0.0016 & 0.0013 & 0.0011 & 0.0009 & 0.0008 & 0.0002   \\\hline
%%$m$               &193    & 228    & 263    & 312    & 361    & 424    & 494    & 578    & 669    & 1016     \\
%\end{tabular}
%\end{center}}
%\end{table}

\begin{table} [!ht]
%\small{
%\tine{
\centering\caption{Estimates of the fractional parameters of the $PM_{10}$.}
\label{Estimatesd1d2PM10}
\begin{center}
\small{
\begin{tabular}{c|c|c|c|c|c|c}  %quantidades de colunas
\hline
Estim.            &\multicolumn{6}{c}{$\alpha$}
                  \\ \cline{2-7}
                  &  0.52  &  0.54  &  0.56   &  0.58  &  0.6   & 0.65(max.)   \\
                  & (361)    & (424)    & (494)    & (578)    & (669)    & (1016)\\
\hline

$\hat{d_1}$          & 0.1004 & 0.1918 & 0.1787 & 0.2071 & 0.1645 & 0.2534   \\
$\hat{d_2}$          & 0.1954 & 0.1798 & 0.1575 & 0.1443 & 0.1548 & 0.0806   \\
Var $\hat{d_1}$      & 0.0110 & 0.0090 & 0.0072 & 0.0060 & 0.0049 & 0.0014   \\
Var $\hat{d_2}$      & 0.0016 & 0.0013 & 0.0011 & 0.0009 & 0.0008 & 0.0002   \\
%$m$               & 361    & 424    & 494    & 578    & 669    & 1016     \\
\hline
\end{tabular}
}
\end{center}
\end{table}

%\begin{table}[!ht]
%	\centering\small
%	\caption{Estimates of $d$ and $D$ for different
%         \textit{bandwidths} ($M=n^{\alpha}$).}
%%	\begin{tabular}{llllll}
%     \begin{tabular}{lllll}
%		\hline
%		 $M$&\multicolumn{1}{c}{$n^{0.5}$}&\multicolumn{1}{c}{$n^{0.6}$}&\multicolumn{1}{c}{$n^{0.7}$}&\multicolumn{1}{c}{$n^{0.8}$}&\multicolumn{1}{c}{$n^{n/2-1}$}\\
%%$M$&\multicolumn{1}{c}{$n^{0.6}$}&\multicolumn{1}{c}{$n^{0.7}$}&\multicolumn{1}{c}{$n^{0.8}$}&\multicolumn{1}{c}{$n^{n/2-1}$}\\		
%\hline
%		%$d$ 	&0.3300 (0.0048)&0.3218 (0.0007)& 0.3207 (0.0019)&0.3274 (0.0007)	 &0.3378 (0.0009)\\
%$d$ 	&0.3218 (0.0007)& 0.3207 (0.0019)&0.3274 (0.0007)	 &0.3378 (0.0009)\\
%%$D$ 	&0.3696 (0.0006)&0.1459 (0.0030)& 0.1116 (0.0007)&0.1164 (0.0013)	 &0.1195 (0.0007)\\
%$D$ 	&0.1459 (0.0030)& 0.1116 (0.0007)&0.1164 (0.0013)	 &0.1195 (0.0007)\\
%		\hline
%	\end{tabular}\label{bandwidth}
%\end{table}

  To obtain the approximation of the model  $\nu_t$
  , the observations were filtered by
  $\nabla^{\hat{\textbf{d}}}$ truncated at $n$=2037. The new series is
  $\hat{\nu}_t=\sum_{j=0}^{n}\hat{\pi}_j^*X_{t-j}$, where
  $\hat{\pi}_j^*$, $j=1,2,\ldots,2037$, are the  estimated coefficients
 of the AR$(\infty)$ representation of a SARFIMA$(0,d_1,0)(0,d_2,0)_7$ model. As an example to verify the impact   of  $X_j$, for large $j$,    in the AR infinite representation, the  $\hat{\pi}_{731}^*$  is $\approx$ $10^{-6}$, which is nearly zero. Since the observations are in scale of $10^{2}$, the contribution of  $X_{j}$  becomes  negligible for large $j$.

Figures \ref{sampleACFPACFARMAPM10}(a) and (b) present the sample
autocorrelation and  partial autocorrelation functions of $\hat{\nu}_t$,
respectively. These plots possibly indicate that an MA(1) model may be
adequate to describe   $\hat{\nu}_t$. However, the MA estimate does not seem to  have significant value. This is in accordance with the stable values of the memory estimates given in  Table \ref{Estimatesd1d2PM10}.

% FIGURES OF THE ACF and PACF of the filtered SERIRES
\begin{figure}[!ht]
\begin{tabular}{cc}
\includegraphics[width=6.0cm,height=4.5cm]{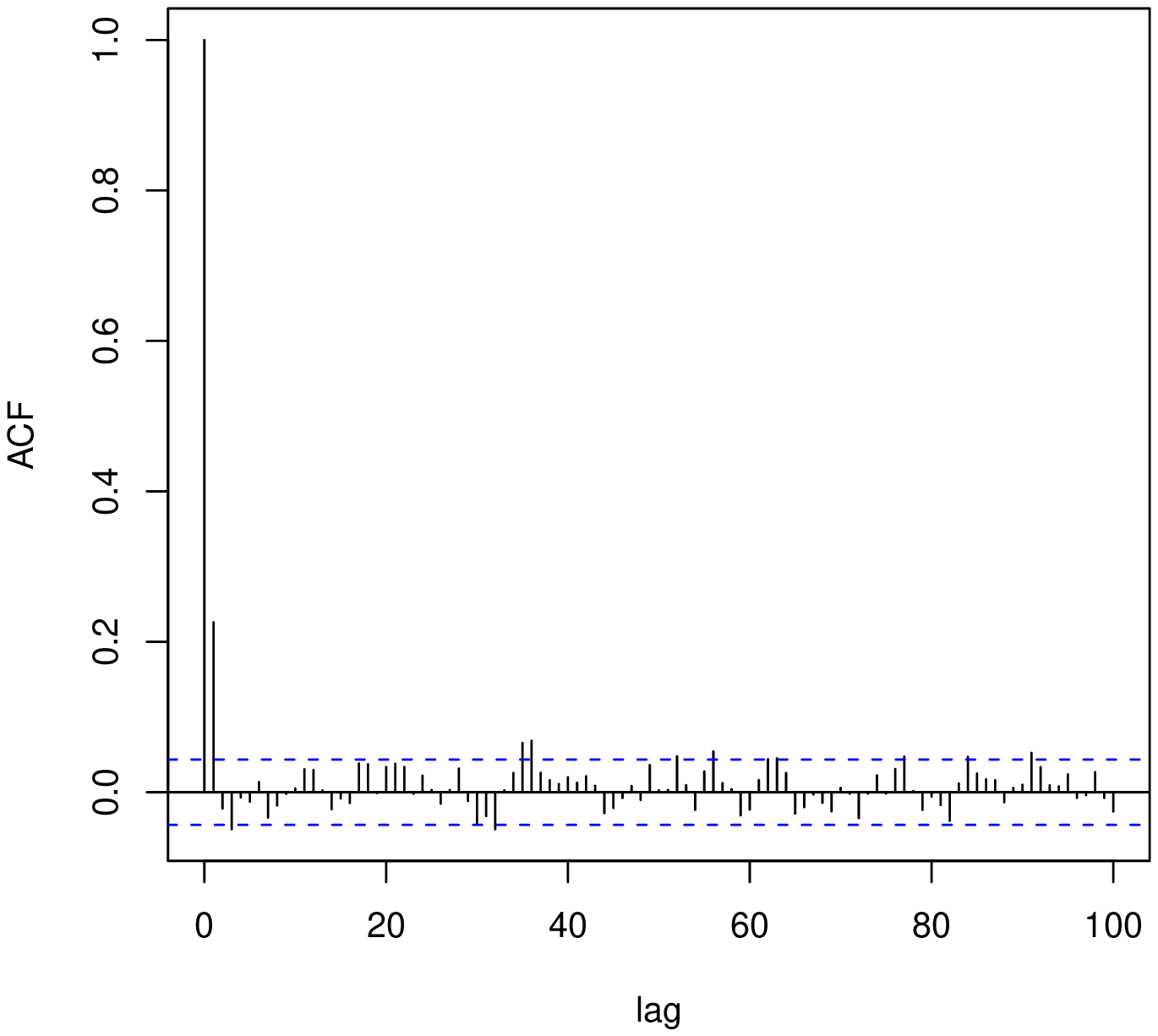}
& \includegraphics[width=6.0cm,height=4.5cm]{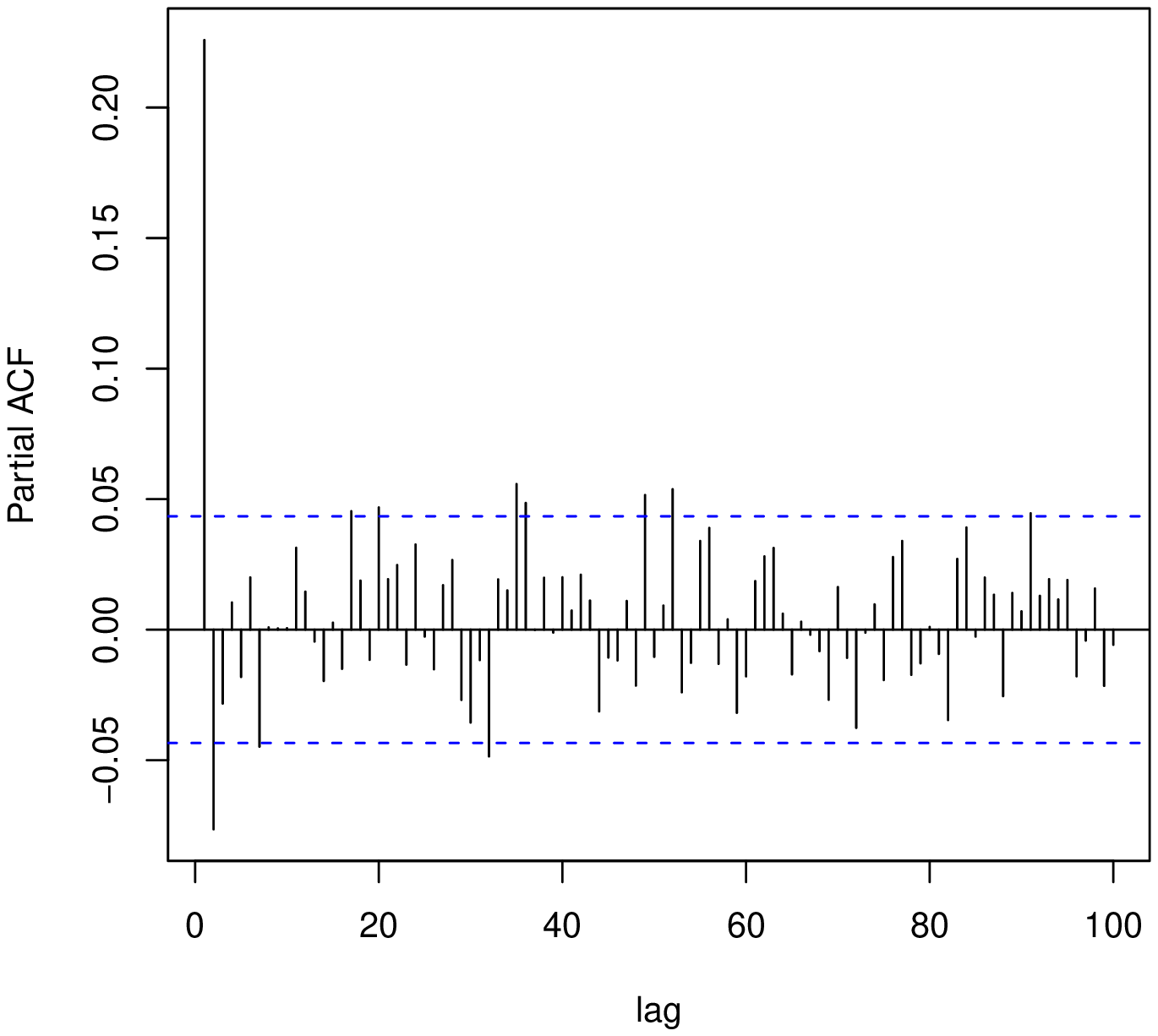}\\
%\includegraphics[width=5.5cm,height=2cm]{Latex/PeACF_m3.pdf}
%& \includegraphics[width=5.5cm,height=2cm]{Latex/PeACF_m4.pdf}
\end{tabular}
\caption{ (a) ACF of $\hat{\nu}_t$ (b)   PACF of $\hat{\nu}_t$.}
\label{sampleACFPACFARMAPM10}
\end{figure}

%% FIGURE3 acf_filtrada
% \begin{figure}[!ht]
%	\centering
%	%	\includegraphics[width=1\textwidth]{acf2.pdf}
%          \includegraphics[width=1\textwidth]{acf_filtrada.pdf}
%	\caption{   Autocorrelation (a) and Partial autocorrelation functions (b) of $\hat{U}_t$.}
%	\label{acf2}
%\end{figure}

To identify the model's order of  $\hat{\nu}_t$, the AIC Criterion  was used which suggested an MA(1) model. Therefore, the  model SARFIMA$(0,d_1,1)\times(0,d_2,0)_7$ with $\hat{\theta} = -0.2673$ ($ sd = 0.0214)$ was chosen for the $PM_{10}$ average data. The standard residual analysis did not present  anomaly of the residuals ( the results are upon-request). Most of the correlations of the residuals fall inside the confidence boundaries ( figures available upon-request).

%\pagebreak

\section{Conclusions}
The paper deals with the seasonal ARFIMA model with two seasonal
fractional parameters. Properties  and model estimation are
discussed. To estimate the parameters,  a multilinear regression
method is   used.   A parametric estimator was also considered for
empirical comparison. The Monte Carlo experiment evidenced that,
in general, all methods gave good estimates and they were very
competitive. The estimators presented very good accuracy for sample size
equal to 1080. The method is very easy to be implemented
and does not require  sophisticated computer capacities. The usefulness  of the SARFIMA model and the semiparmetric fractional estimator was exemplified using  artificial  and a  daily average PM$_{10}$ concentration series.

\section*{Acknowledgements}
V. A. Reisen and B. Zamprogno gratefully acknowledge partial financial support from PIBIC-UFES and
CNPq/Bra\-zil. W. Palma was partially supported by Fondecyt grant number
 1085239. J. Arteche acknowledges financial support from the Spanish Ministerio de Ciencia y Tecnología
 and ERDF grant SEJ2007-61362/ECON. The authors thank  Prof. Liudas Giraitis for useful suggestions on the  theory of the seasonal model.  Part of this work was done during the period that V.  A. Reisen was visiting Centre
 D'Economie de la Sorbonne, Paris. He thanks Prof. D. Guegan for her very kind invitation.

\newpage
\section*{APPENDIX A}
{\bf Proof of Theorem 1:} The asymptotic mean and variance of $\hat{d}$ follows as in Theorem 1 in Hurvich et al (1998). Denote $Z$ the $\sum_km_k\times 2$ matrix with the regressors in (\ref{model9}) and similarly the vector  $V$ for the disturbances. Then
\[\hat{d}-d=(Z'Z)^{-1}Z'V.\]
Note now that
\begin{eqnarray*}
2\log X_{1,k,j}&=&2\log\left\{2 \sin\left(\pi k+\frac{\pi j s'}{n}\right)\right\}=2\log|\la_{js'}|+ O(|\la_{js'}|^2) \\
2\log X_{2,k,j}&=&2\log\left\{2 \sin\left(\frac{\pi k s_2}{s'}+\frac{\pi j s_2}{n}\right)\right\}=2\log|\la_{js_2}|+ O(|\la_{js'}|^2)
\end{eqnarray*}
if $k\in I_k$ and
\begin{eqnarray}
2\log X_{2,k,j}&=&2\log \left\{2\sin\left(\frac{\pi k s_2}{s'}\right)\left[1+\frac{\cos(\pi k s_2/s')}{\sin(\pi k s_2/s')}
\frac{\pi |j| s_2}{n}+O(\la_{js_2}^2)\right]\right\} \label{eq7}
\end{eqnarray}
if $k\not\in I_k$. Then
\begin{eqnarray*}
\sum_{k=0}^{[s'/2]}\sum_j^*Z_{1,k,j}^2&=&4\sum_{k=0}^{[s'/2]}m_k (1+o(1)) \\
\sum_{k=0}^{[s'/2]}\sum_j^*Z_{2,k,j}^2&=&4\sum_{k\in I_k}m_k (1+o(1))\\
\sum_{k=0}^{[s'/2]}\sum_j^*Z_{1,k,j}Z_{2,k,j}&=&4\sum_{k\in I_k}m_k (1+o(1))
\end{eqnarray*}
This result follows from the fact that for those $k\not\in I_k$, $\sum_j^*Z_{2,k,j}^2=O(m^3n^{-2})=o(m)$ and $\sum_j^*Z_{2,k,j}Z_{1,k,j}=O(m^2n^{-1}\log m)=o(m)$. Then
\begin{equation}
Z'Z=mQ(1+o(1))\label{eq5}
\end{equation}

Denoting now $\ep$ the vector with elements $\ep_{k,j}$ we have that
\begin{equation}
Z'\ep =mb_n (1+o(1))\label{eq6}
\end{equation}
since
\begin{eqnarray*}
\sum_{k=0}^{[s'/2]}\sum_j^*Z_{1,k,j}\ep_{k,j}&=&-2m\sum_{k=0}^{[s'/2]}b_k\delta_k(2\pi)^{\al_k}\frac{\al_k}{(\al_k+1)^2}
\left(\frac{m}{n}\right)^{\al_k}\left(1+O\left[\log m\left(\frac{m}{n}\right)^{\iota}\right]\right) \\
\sum_{k=0}^{[s'/2]}\sum_j^*Z_{2,k,j}\ep_{k,j}&=&-2m\sum_{k\in I_k}b_k\delta_k(2\pi)^{\al_k}\frac{\al_k}{(\al_k+1)^2}
\left(\frac{m}{n}\right)^{\al_k}\left(1+O\left[\log m\left(\frac{m}{n}\right)^{\iota}\right]\right)
\end{eqnarray*}
because for $k\not\in I_k$
\begin{eqnarray*}
\sum_{k\not\in I_k}\sum_j^*Z_{2,k,j}\ep_{k,j}&=&O\left(\sum_j^*|\la_j|^{\al_k+1}\right)=o\left(m\left[\frac m n\right]^{\alpha_k}\right)
\end{eqnarray*}
The rest of the proof follows as in Hurvich et al (1998).

{\bf Proof of Theorem 2:} The proof follows as in Hurvich et al (1998) applying Lemma 4 in Sun and Phillips (2003) which holds in our multiple log periodogram regression context. Since
\[\sqrt{m}(\hat{d}-d)=(m^{-1}Z'Z)^{-1}m^{-1/2}Z'\ep+(m^{-1}Z'Z)^{-1}m^{-1/2}Z'U\]
by using (\ref{eq5}) and (\ref{eq6}) it only remains to show that $m^{-1/2}v'Z'U\stackrel{d}{\ri} N(0, \pi^2 v'Qv/6)$ for any vector $v=(v_1,v_2)$. As in Hurvich et al. (1998)
\[\frac{1}{\sqrt{m}}v'Z'U=o_p(1)+\frac{1}{\sqrt{m}}\sum_{k=0}^{[s'/2]}\sum_{|j|>m^{0.5+\delta}}^{*}g_{k,j}U_{k,j}\]
for some $0.5>\delta>0$ and $g_{k,j}=v_1Z_{1,k,j}+v_2Z_{2,k,j}$. Now $\max_{j,k}|g_{k,j}|=O(\log m)$ and
$\sum_{|j|>m^{0.5+\delta}}^{*}|g_{k,j}|^p=O(m)$ for all $p\geq 1$ (see formula (A18) in Hurvich et al (1998) for
$Z_{1,k,j}$ and $Z_{2,k,j}$ for $ks_2\in I_k$ and use (5) for $ks_2\not \in I_k$). Since by equation (\ref{eq5})
$\sum_{|j|>m^{0.5+\delta}}^{*}g_{k,j}^2=mv'Qv(1+o(1))$, we can apply Lemma 4 in Sun and Phillips (2003) to get the desired result.

\end{document}